\documentclass[12pt]{article}

\usepackage[utf8]{inputenc}
\usepackage{graphicx}
\usepackage{amsmath,amssymb,amsthm}

\usepackage{xcolor}
\definecolor{linkcolor}{RGB}{0, 0, 255}
\definecolor{citecolor}{RGB}{0, 0, 255}
\definecolor{urlcolor}{RGB}{0, 0, 255}
\usepackage[
  colorlinks=true,          
  linkcolor=linkcolor,      
  citecolor=citecolor,      
  urlcolor=urlcolor         
]{hyperref}

\usepackage{geometry}
\usepackage{fancyhdr}
\usepackage[natbibapa]{apacite}

\setcitestyle{aysep={}}

\newtheorem{Proposition}{Proposition}[section]
\newtheorem{Definition}{Definition}[section]
\newtheorem{Remark}{Remark}[section]

\title{The Causal Second Law}
\author{Bal\'azs Gyenis}
\date{Institute of Philosophy, ELTE RCH, Budapest, Hungary \\
\medskip
\href{mailto:gyepi@hps.elte.hu}{gyepi@hps.elte.hu} \\
\medskip
{\bf Received}: 18 February 2025 | {\bf Revised}: 5 February 2026 \\ {\bf Accepted}: 17 February 2026 | {\bf Online}: 30 April 2026 \\ 
\bigskip
\bigskip
\href{https://onlinelibrary.wiley.com/share/author/NHAHRIVR9BHQUFXEPAXK?target=10.1111/nous.70042}{ePDF} \\
\bigskip
\bigskip
DOI: \href{https://doi.org/10.1111/nous.70042}{10.1111/nous.70042}
}

\begin{document}



\maketitle

\begin{abstract}
I argue that if a special science satisfies certain key assumptions that are familiar from physicalist accounts of the special sciences and from physics, then its causal regularities have an associated notion of entropy, and that this causal entropy cannot decrease from a robust cause to its effect. Due to its analogy with the second laws of thermodynamics and statistical physics, I call the latter conclusion the causal second law. In this paper, I clarify the key assumptions, prove the causal second law, give sufficient conditions for causal entropy increase, relate the causal second law to statistical mechanics and thermodynamics, and argue that the reversibility objection does not threaten it. In addition, I claim that the causal second law is compatible with a non-metaphysical understanding of supervenience and the open systems view, argue that it does not imply a causal time arrow, reflect on relaxing the robustness condition, question whether it is necessary to invoke thermodynamics to show that special sciences' time arrows exist, and discuss a transition-relative-frequency-based, special-science-internal characterization of causal regularities.  
\end{abstract}


\section{Introduction}

Causal regularities have an associated notion of entropy which cannot decrease from a robust cause to its effect, assuming that certain key assumptions, familiar from physicalist accounts of the special sciences and from physics, hold. In the limiting case where, as \citet[173]{Hume1739}  puts it, ``the same cause always produces the same effect,'' an argument for this claim can be succinctly formulated as follows. If the same cause always produces the same effect, then any physical instantiation of the cause, following the laws of physics, must evolve to a physical instantiation of the effect. But then, if the number of distinct physical instantiations of the cause cannot decrease during their time evolution, there must be at least as many physical instantiations of the effect as physical instantiations of the cause. By defining the causal entropy of the effect and the cause as the number of their physical instantiations, it follows that the causal entropy of the effect must be at least as large as that of the cause.

The paper clarifies and motivates the key assumptions---primarily, state-supervenience and measure-preservation---behind a physically more plausible version of this argument and extends the argument to robust causal regularities of any special science (which term I understand here broadly to include folk physics and folk psychology, along with chemistry, biology, psychology, economics, and so on) for which these assumptions hold. With this terminology, the main claim can be restated as follows: if a special science has robust causal regularities, and if it satisfies state-supervenience and measure-preservation, then the special science has an entropy principle tied to its own domain. Due to its analogy with the second laws of thermodynamics and statistical physics, I call this entropy principle the causal second law.\footnote{For a discussion of entropy principles and various formulations of the second law, see \citet{Uffink2001}.}

Section \ref{sect_causalsecondlaw_robust} details the assumptions and proves the causal second law for robust causal regularities on their basis. In Section \ref{sect_strictincrease_robust}, I argue that when an effect has multiple possible causes, a case that is typical in special sciences, causal entropy from an actual robust cause to its effect strictly increases. I also show why causal entropy must generally increase due to the mismatch between the descriptive capabilities of a given special science and physics.

In Section \ref{sect_thermo}, I turn to the relationship between the causal second law and the second laws of thermodynamics and statistical mechanics. I reconstruct an explication of the second law by \citet{Jaynes1965}, arguing that it provides a perspective in which thermodynamic entropy can be understood as causal entropy applied to the special science of thermodynamics. This justifies introducing the term entropy outside thermodynamics and statistical mechanics, associating it generally with causes and effects, and illustrates that, under certain circumstances, practically useful expressions can be derived from the causal second law. I also address time reversal invariance, showing that the reversibility objection does not threaten the exceptionless character of the causal second law for robust causal regularities of special sciences that satisfy the key assumptions.

In Section \ref{sect_motivatingtheassumptions}, I return to the background assumptions; in particular, I argue that the approach can be extended to multiply realized robust causation, that we can relax the metaphysical reading of state-supervenience, and that the relevant sense in which measure-preservation is invoked is compatible with the open systems view. Although the main claims of this paper remain conditional upon the so-relaxed assumptions, the section briefly motivates their general plausibility. I argue that the reasons underlying the strict increase of causal entropy for robust causes also allow robustness to be relaxed to portionality.

Finally, Section \ref{sect_dynamicalsystemsapproach} is devoted to the framework of the discussion, the dynamical systems approach to causation. I point toward a number of open problems and directions for further development. I argue that the approach provides a clear distinction between entropy-from-cause-to-effect and entropy-in-time (Section \ref{subsect_fourcausalasymmetries}); emphasize the description-relativity of the results (Section \ref{subsect_howspecialsciencedescriptionscomeabout}); address the viability of philosophical projects that aim to infer that causes precede their effects from an entropy increase in thermodynamics (Section \ref{subsect_rovelli}); ask whether, in the light of the results of this paper, it is necessary to invoke thermodynamics to show that special sciences' time arrows exist (Section \ref{subsect_thermoreduction}); show that strengthening state-supervenience to history-supervenience permits a special-science-internal characterization of robust regularities (Section \ref{subsect_internalcharacterization}); and address the practical usefulness of the causal second law (Section \ref{subsect_usefulness}).

The main text keeps the discussion informal; formal definitions and proofs are in the Appendix.

\section{The Causal Second Law}\label{sect_causalsecondlaw_robust}

Even though David Hume's attempt to {\em define} the causal relation from regularities has been frequently criticized \citep{Andreas-Guenther2021}, the {\em property} that a cause is regularly followed by its effect is one of the key characteristics of causal claims in the special sciences and everyday life. In this paper, I focus on causal claims that can be reconstructed as stating that a conjunction of causal factors in a situation regularly leads to an effect in a time characteristic of the regularity. For brevity, I will refer to such a conjunction of causal factors as a cause (for important qualifications, see Section \ref{sect_dynamicalsystemsapproach}).

\begin{figure*}
	\begin{center}
		\includegraphics[width=7.5cm]{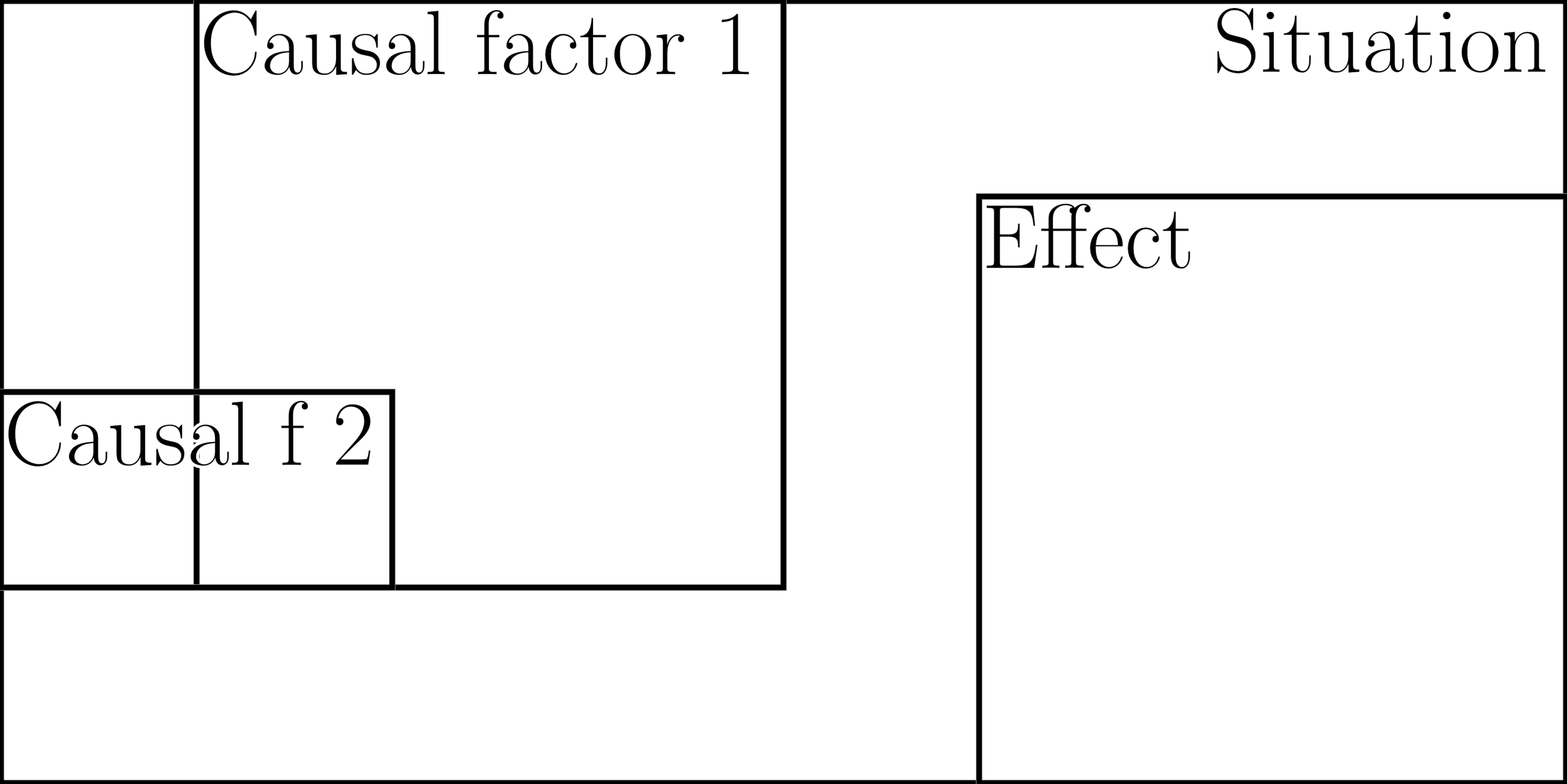}
		\includegraphics[width=7.5cm]{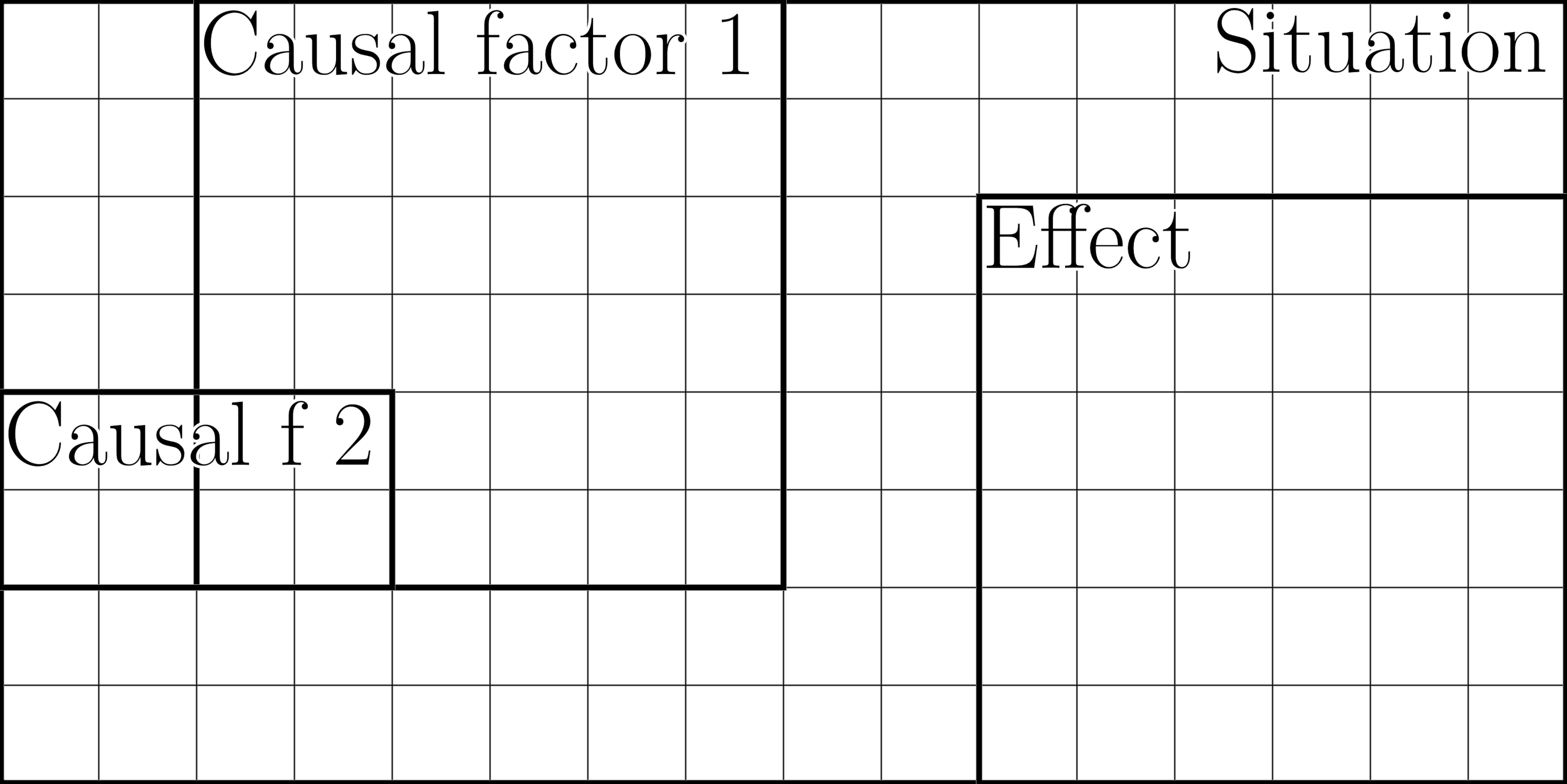}
		\caption{{\em Left: an Euler diagram of special science states of affairs: a situation, two causal factors, and an effect. Right: state-supervenience: correspondence of special science descriptions to physical states.}}
		\label{figure_01}
	\end{center}
\end{figure*}

Examples of such causal claims abound. A large dose of penicillin and a severe penicillin allergy in a patient regularly lead to the patient having gone into anaphylactic shock in ten minutes (medicine). Sleep deprivation and prolonged social isolation in young adults regularly lead to mood having become dysregulated in 2 days (psychology). Printing large amounts of currency and keeping interest rates near zero in a modern economy regularly lead to inflation having significantly risen in 3 months (economics). Rapid tectonic plate movement and locked fault zones in a seismic region regularly lead to a major earthquake having occurred in 30 years (geology). 

The causal second law for robust regularities rests on two key assumptions, which I shall refer to as state-supervenience and measure-preservation. Here I state and illustrate these assumptions; I postpone discussing and relaxing them until Section \ref{sect_motivatingtheassumptions}.

First, I assume what I call {\em state-supervenience} of a given special science on physics: that every description of states of affairs of the special science corresponds to a set of physical states that instantiate said special science description. Consider the folk physics example that a burning match and a large amount of flammable material on the floor of a typical house regularly leads to a burned down house. According to the assumption, there exists a set of physical states that correspond to the folk physics description that `there is a burning match on the floor of a house'. Whenever this folk physics description is true in a particular place and time, it is instantiated by one of the physical states in this set (i.e., by the physical state that specifies the way particles and fields making up the burning match and the house are distributed in that particular place and time). Similarly, there exists a second set of physical states that instantiate flammable material on the floor of a house, and a third set of physical states that instantiate a burned down house. The right side of Figure \ref{figure_01} illustrates the idea by representing a physical state as a cell of a grid, and special science descriptions of states of affairs as cells enclosed within boldly drawn boundaries.

\begin{figure*}
	\begin{center}
		\includegraphics[width=7.5cm]{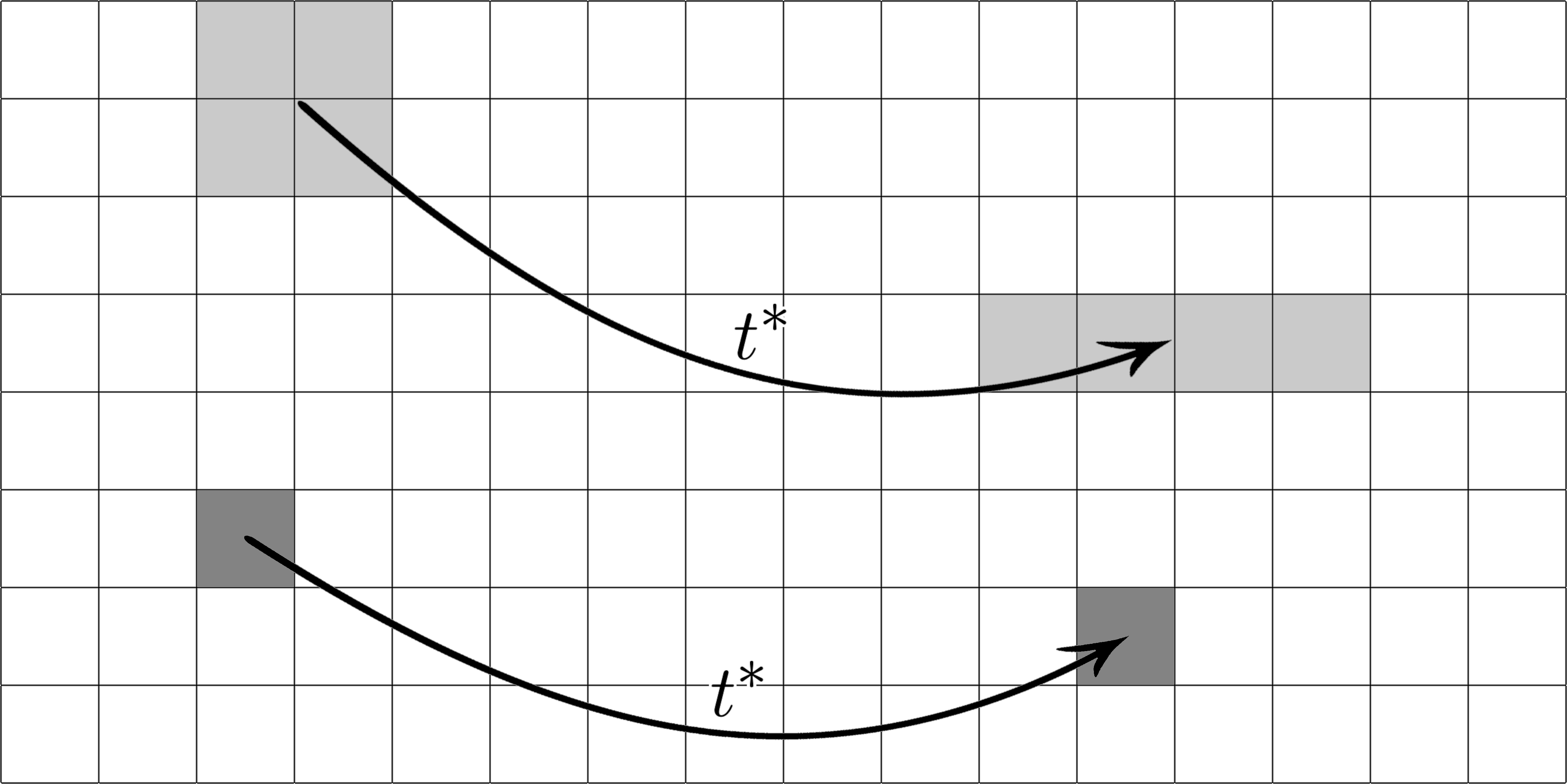}
		\includegraphics[width=7.5cm]{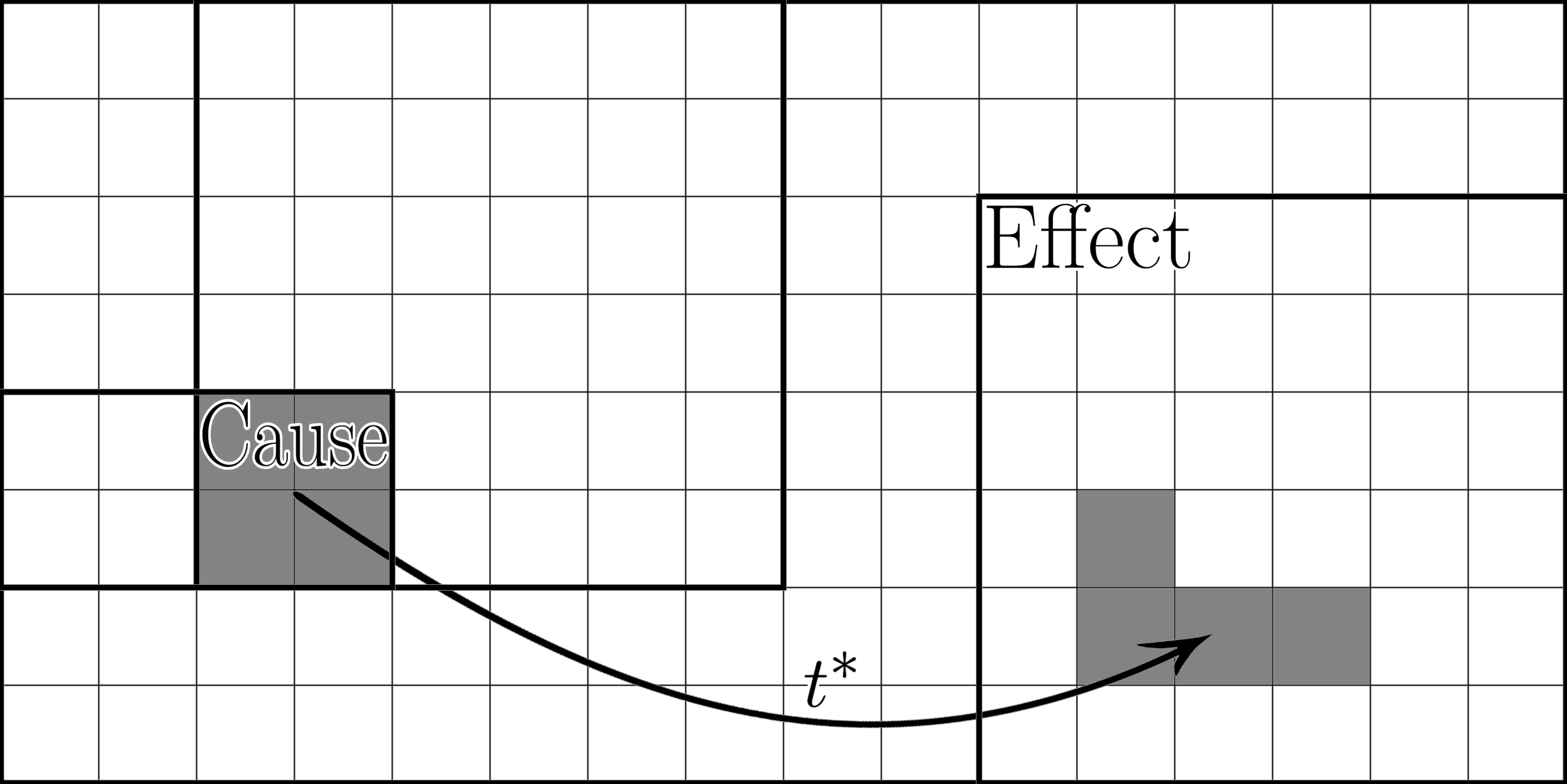}
		\caption{{\em Left: time evolution of one physical state, and (joint) time evolution of a set of physical states. Right: a robust causal regularity, defined with respect to a dynamical system, connects sets of physical states.}}
		\label{figure_02}
	\end{center}
\end{figure*}

Since a special science description corresponds to a set of physical states, the change of a special science description in time depends on how the laws govern the time evolution of the physical states that instantiate it. For simplicity, I assume that the laws of the underlying physical theory are deterministic, and that a physical state belongs to the dynamical state space of the theory, to which I will succinctly refer as {\em phase space} (cf. footnote \ref{footnote_bohmian} on terminology; see Section \ref{sect_motivatingtheassumptions} for further discussion). The time evolution of physical states is illustrated on the left side of Figure \ref{figure_02}: in a given time, a single arrow from one cell to another represents the evolution of one physical state into another, and a single arrow from one set of cells to another represents the evolution of physical states in the first set into states in the second set.

Suppose that, at a particular place and time, there is a burning match and flammable material in a house. This state of affairs leads to a burned down house only if the particular physical instantiation of the cause (i.e., the actual way particles and fields constitute the burning match and the flammable material in the house at this particular place and time), following the laws of physics, evolve to a particular physical instantiation of the effect (i.e., the actual way particles and fields constitute the particular burned-down house). In different places and times, the same cause may be instantiated by different physical states. However, if the cause regularly leads to the effect, then what we assumed for a particular place and time must hold in general: physical states instantiating the presence of a burning match and flammable material in a house must regularly evolve to physical instantiations of a burned-down house.

I call a causal regularity {\em robust} if almost all physical states that instantiate the cause evolve, in a time characteristic of the regularity, to physical states instantiating the effect (Definition \ref{def_robust}, Figure \ref{figure_02}). I will refer to the cause of a robust causal regularity as a {\em robust cause}. Since the causal second law has its cleanest formulation for robust regularities, most of the paper focuses on them; however, Section \ref{sect_motivatingtheassumptions} relaxes robustness to portionality. Section \ref{subsect_internalcharacterization} further shows how strengthening state-supervenience to history-supervenience allows for a special-science-internal characterization of causal regularities that does not invoke instantiation by physical states.

The Appendix contains a mathematical characterization of these informally expressed assumptions using the theory of dynamical systems (Definition \ref{def_dynamicalsystem}). The {\em number} of states instantiating a special science description may be infinite, making it difficult to compare this number for one description with that for another. This technical difficulty is readily resolved by the introduction of a {\em measure} of sets of physical states; intuitively, measure generalizes state counting. Invoking measure theory also allows precise characterization of claims invoking ``almost all'' and ``almost always'': these terms refer to conditions where exceptions form a measure zero set. 

There are multiple ways to generalize state counting to the infinite case, since different mathematical measures can be introduced on a set of physical states. However, my focus here is on a special type of measure, which expresses the idea that the number of physical states remains unchanging as time passes: a dynamical system is called {\em measure-preserving} if the time evolution does not change the measure of any (measurable) set of physical states (Definition \ref{def_measurepreserving}). 

A key physical observation is that Liouville's well-known theorem entails that conservative dynamics of all accepted physical theories that have a Hamiltonian formulation are measure-preserving with respect to the so-called symplectic phase-volume measure. The list includes classical mechanics, statistical mechanics, electrodynamics, and non-exotic models of general relativity among others. In the rest of this paper, I will succinctly refer to the symplectic phase-volume measure (or to its theory-appropriate generalization\footnote{Bohmian quantum mechanics fits Definition \ref{def_measurepreserving} of a measure-preserving dynamical system with the dynamical state space $P$ that is the Cartesian product of the projective Hilbert space and the particle configuration space (however, since this is not a Hamiltonian formulation, strictly speaking the so-defined dynamical state space $P$ should not be referred to as ``phase space'' and the invariant measure as a ``phase volume''). Therefore, mathematical results of this paper generalize to Bohmian quantum mechanics. For quantum theories with infinite degrees of freedom, other analogous notions of physical state counting may apply. Extending the arguments of this paper to theories that cannot be cast as a measure-preserving dynamical system shall be the focus of future research. \label{footnote_bohmian}}) as {\em phase volume}. 

With these technical remarks in mind, I define the {\em causal entropy} of a cause and of an effect as the phase volume of the set of physical states instantiating each.\footnote{In certain contexts it may be convenient to define causal entropy as a monotonic function of the phase volume; notably, if causal entropy is defined as the scaled logarithm of the phase volume, then it becomes a generalization of the so-called Boltzmann entropy of statistical mechanics to any special science satisfying state-supervenience; see Section \ref{sect_thermo}.} 

State-supervenience and measure-preservation immediately yield (Proposition \ref{prop_causalsecondlaw}):
\begin{description}
	\item The {\bf (dynamical) causal second law} (for robust causal regularities): causal entropy cannot decrease from a robust cause to its effect.
\end{description}

The proof can be followed on the right side of Figure \ref{figure_02}. If the causal regularity from a cause to an effect is robust, then almost all physical states instantiating the cause evolve, in a characteristic time, to physical states instantiating the effect. But if the dynamics is measure-preserving, then the phase volume of the effect must be at least as large as that of the cause to accommodate the physical states arriving from the cause, which maintain the same phase volume during their time evolution. Thus, the causal entropy of a robust cause cannot be larger than the causal entropy of its effect.\footnote{{\em Mathematically}, the proof is straightforward. However, the literature usually invokes measure-preservation to argue {\em against} the very possibility of a dynamical systems-based derivation of a common formulation of the second law of thermodynamics (a formulation that implies that once thermodynamic equilibrium is reached it is never left again; see, e.g., \citealt[59--60]{Hemmo-Shenker2012}). In this context, the fact that measure-preservation becomes the key assumption {\em entailing} the causal second law may not be {\em conceptually} straightforward. (The two arguments are compatible because thermodynamic equilibrium corresponds to the macrostate with the largest phase volume, hence, under measure-preservation, it can only be the {\em effect} state of a robust causal regularity, but it cannot be a robust {\em cause} state of another effect. Thus, the causal second law for robust causal regularities does not prohibit fluctuation away from thermodynamic equilibrium. See also footnote \ref{footnote_boltzmannian}.) \label{footnote_comparetwolaws}}

\section{The Increase of Causal Entropy}\label{sect_strictincrease_robust}

When the causal entropy stays the same from a robust cause $C$ to an effect $E$, it may be that $E$ is also a robust cause of $C$, as in the case of an oscillating phenomenon. When causal entropy strictly increases, measure preservation entails that $E$ cannot be a robust cause of $C$, accounting for a causal asymmetry between robust causes and their effects (for further discussion, see Section \ref{subsect_fourcausalasymmetries}). 

A strict increase of causal entropy also allows relaxing robustness (Section \ref{subsect_relaxingrobustness}). This section gives two sufficient conditions for strict causal entropy increase.

\subsection{Multiplicity of Possible Causes}\label{subsect_multiple}

Special science and everyday causal claims often leave effects underspecified, allowing multiple distinct possible causes. A burning match in the presence of flammable material in a house (cause 1) regularly leads to a burned-down house (effect). But a lit lighter in the presence of flammable material in a house (cause 2) also regularly leads to a burned-down house (same effect), as does a spark from an exposed electrical circuit or an operating flamethrower (Figure \ref{figure_03}).

\begin{figure*}
	\begin{center}
		\includegraphics[width=7.5cm]{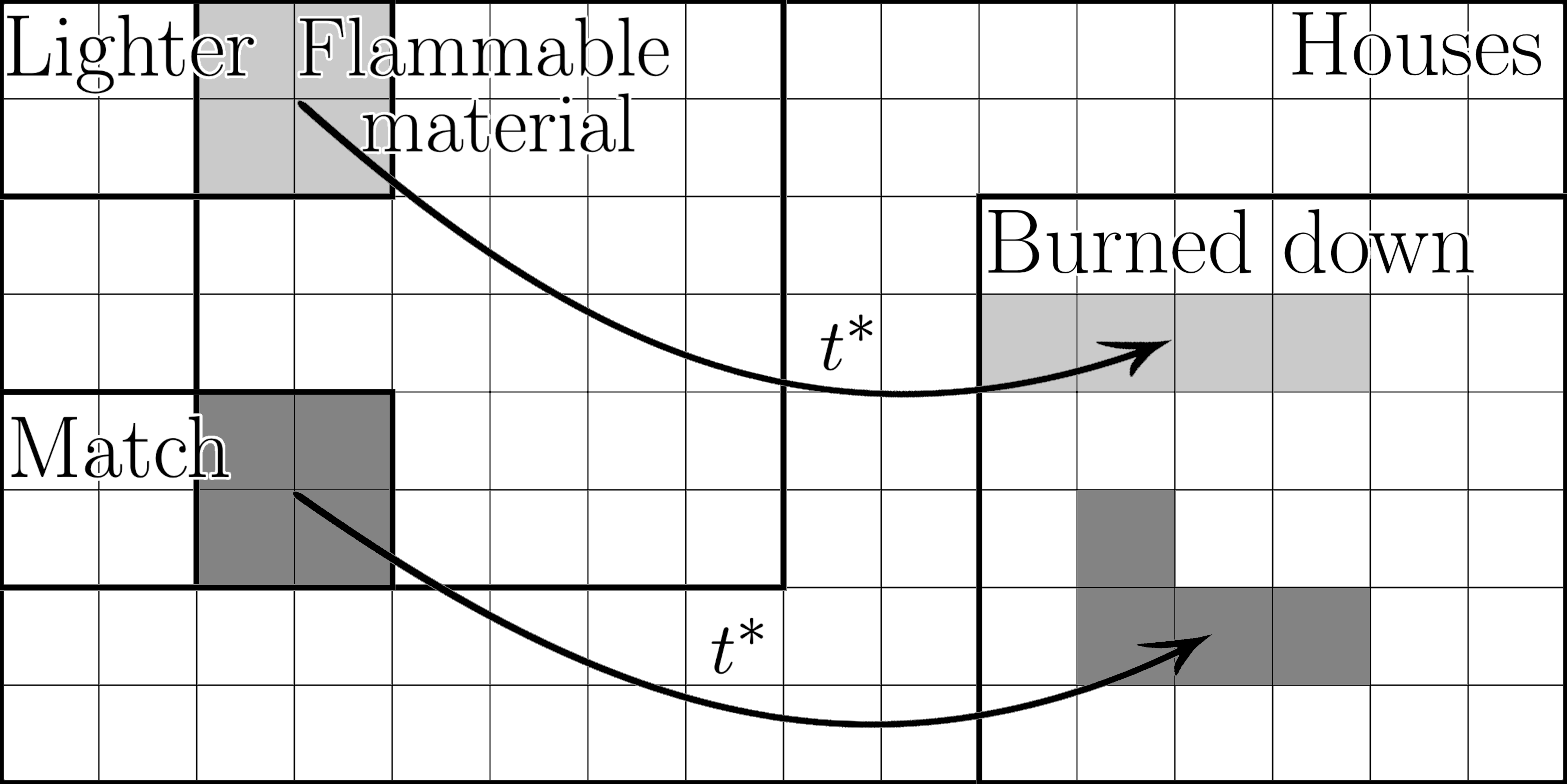}
		\caption{{\em If there are multiple distinct possible causes which robustly lead to the same effect, then the phase volume of the effect is larger than the phase volume of any of the possible causes.}}
		\label{figure_03}
	\end{center}
\end{figure*}

If there are multiple distinct {\em possible} causes that robustly lead to the same effect, then the causal entropy in any {\em actual} robust causal process must increase from the cause to the effect. Following the argument of the previous section, the entropy of the effect is at least as large as the entropy of the disjunction of all the possible causes. Since there are several distinct possible causes, the entropy of the disjunction of all possible causes is greater than that of any possible cause by itself. Hence, for any actual robust causal process, the entropy of the effect is larger than the entropy of the actual cause (Proposition \ref{prop_robust_multiple}).\footnote{Entropy increase due to multiple distinct possible causes also holds if their characteristic times differ, under a condition that can be verified by the special sciences. Suppose the characteristic time of a burning match in the presence of flammable material (but without an operating flamethrower) robustly leading to a burned down house is 1 h, but an operating flamethrower in the presence of flammable material (but without a burning match) robustly leads to a burned-down house in 50 min.  The condition for entropy increase in Proposition \ref{prop_robust_multiple} demands that there are at least some cases in which, when a fire is started by a burning match (but without the presence of an operating flamethrower), we do not find ourselves 10 min later without either a fire or a burning match, only with an operating flamethrower that has just started the fire.}

\subsection{Mismatching Descriptive Capabilities}\label{subsect_mismatching}

One may argue that the result of causal entropy increase in the previous subsection depended on not considering a sufficiently wide phase space region---namely, the set of all physical states that lead, in the same characteristic time, to the effect (the ``pull-back'' of the effect)---as the most general cause of the effect. However, nothing guarantees that this wide phase space region corresponds to any state of affairs that is describable by the special science in question.

To illustrate with a simple example, assume that the only descriptions of states of affairs expressible by a given special science correspond to the regions enclosed within boldly drawn boundaries in Figure \ref{figure_04}. Thus, regions that can be described by the special science include three causal factors, two possible causes (the $2 \times 2$ cell regions colored with lighter and darker grey on the right side of Figure \ref{figure_04}), the effect (the $6 \times 6$ cell region colored with various shades of grey on the left side of Figure \ref{figure_04}), and some other descriptions such as `either one or the other possible cause is present,' `neither of the causal factors nor the effect is present,' etc. However, in our example, the pull-back of the effect (the entire curvy, gray-colored $36$-cell region on the right side) cannot be described by the special science as a state of affairs. If the special science wanted to assert a state of affairs that covers all cases in which the pulled-back effect obtains, the only describable states of affairs it has at its disposal would correspond to a region that properly includes the pulled-back effect. In Figure \ref{figure_04}, the smallest describable region that properly includes the pull-back of the effect is the completely vague description that something occurs; however, this description is not a robust cause of the effect.

Thus, as Figure \ref{figure_04} illustrates, for a given effect there may be no single state of affairs expressible by the special science that can both include all physical states leading to the effect and also serve as a robust cause for it. Such a mismatch between the descriptive capabilities of a special science and physics then entails a strict increase in causal entropy from any expressible actual cause to the effect (Proposition \ref{prop_robust_mismatch}).

\begin{figure*}
	\begin{center}
		\includegraphics[width=7.5cm]{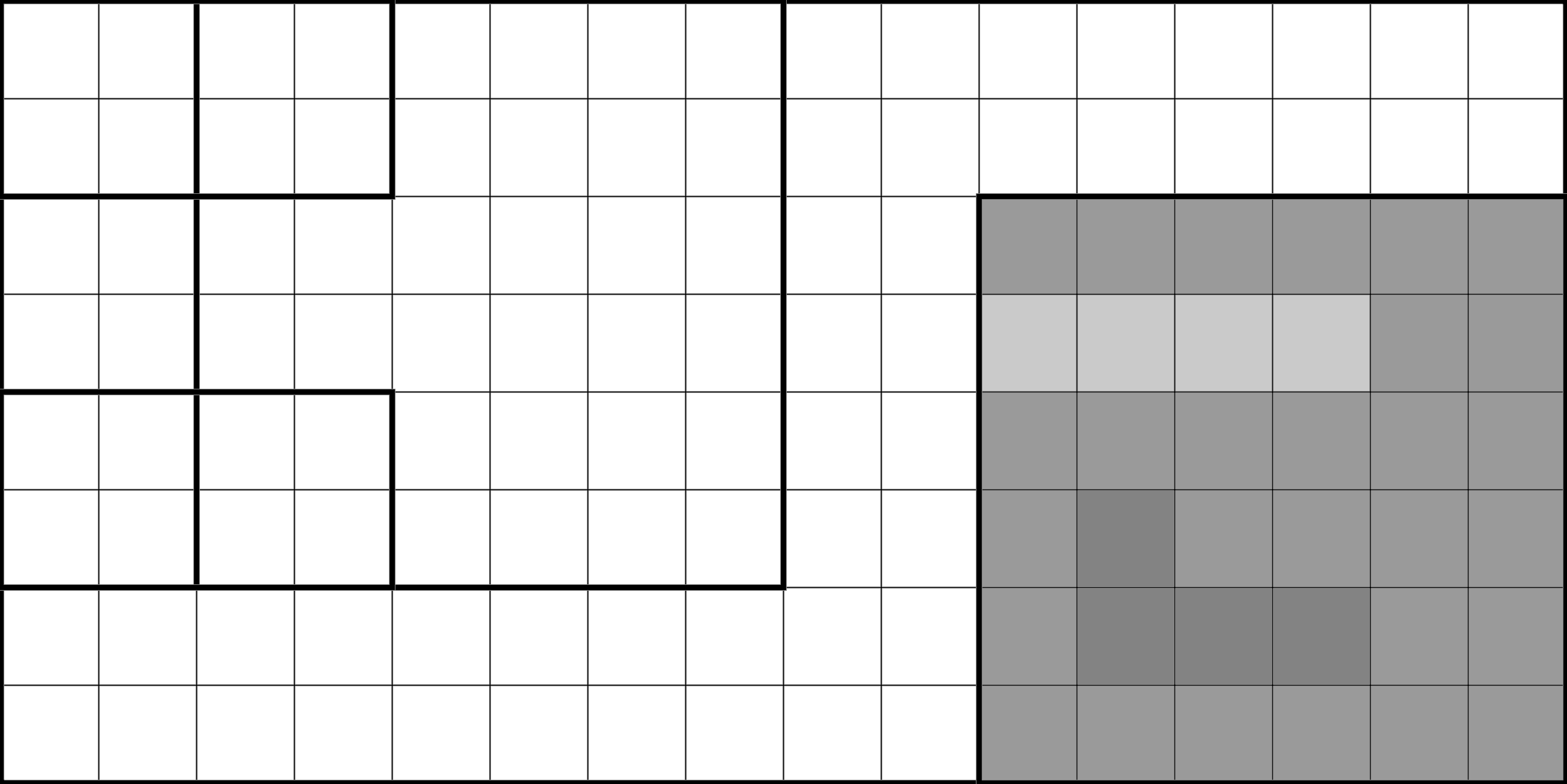}
		\includegraphics[width=7.5cm]{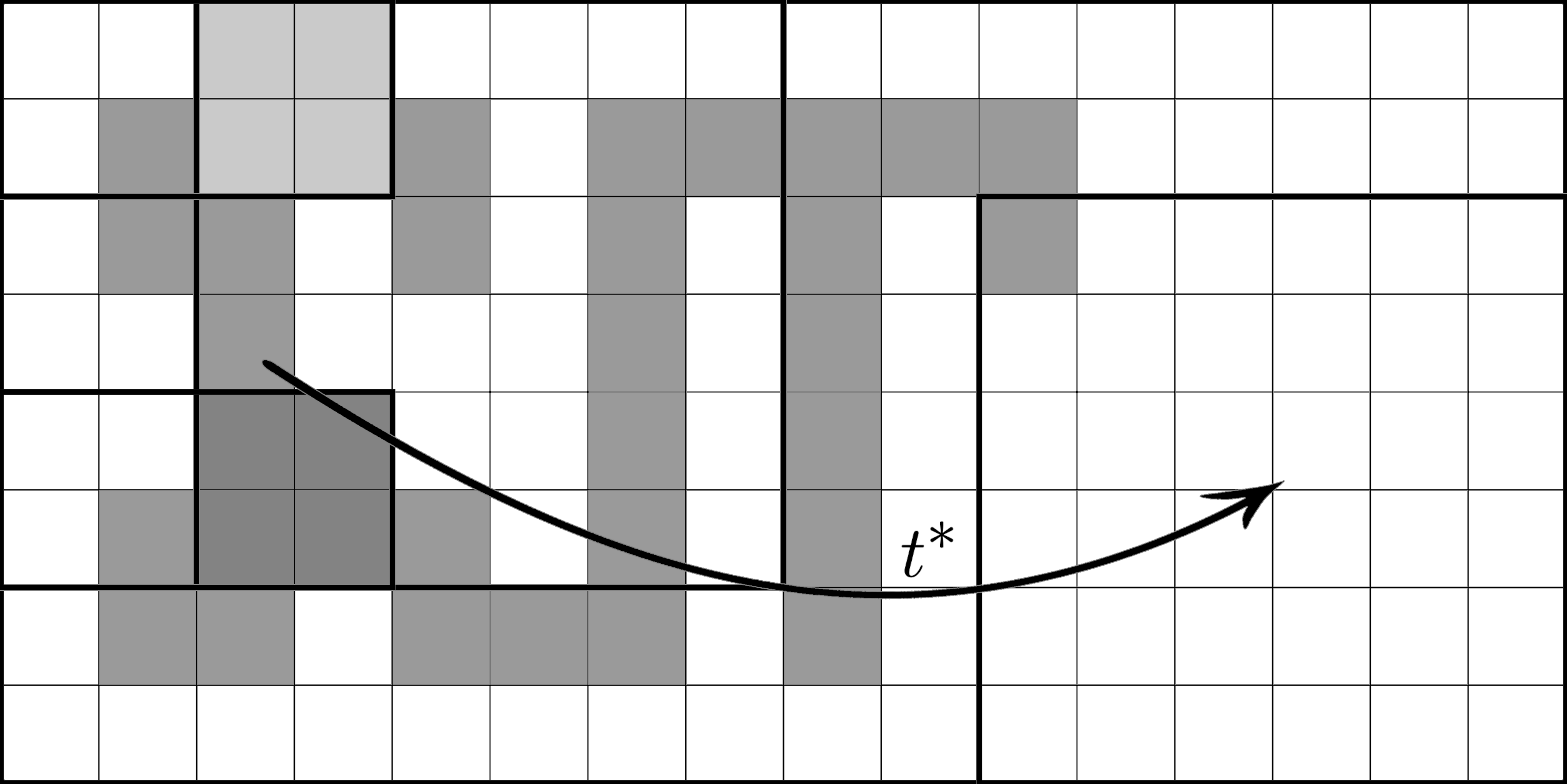}
		\caption{{\em Left: the set of all physical states instantiating the effect. Right: the set of all physical states that evolve, in a characteristic time $t^*$, to the effect.}}
		\label{figure_04}
	\end{center}
\end{figure*}

The increase in causal entropy from cause to effect due to the mismatch is {\em relative} to the descriptive capabilities of the special science and the underlying physical theory; however, given these descriptive capabilities, it is {\em objective} in the sense that it does not depend on the assumption that agents, their subjective beliefs, or their more-or-less coarse-grained informational states exist. 

Descriptions of different special sciences partition the same phase space in different ways (see Section \ref{subsect_howspecialsciencedescriptionscomeabout}). For different special sciences, the entropy increase due to the mismatch may thus be different; for instance, it is less for a special science that provides more refined descriptions of states of affairs. In the limit where the descriptive capability of a special science matches that of the underlying physical theory (when the regions enclosed within bold boundaries narrow down to the cells), robust causation reduces to determinism, and the increase in causal entropy from cause to effect due to the mismatch disappears.

\section{The Causal and the Thermodynamic Second Laws}\label{sect_thermo}

\subsection{Jaynes' Explication of the Phenomenological Second Law}\label{subsect_jaynes}

This section shows, in two steps, that by following a reasoning of \citet{Jaynes1965} in reverse, the thermodynamic (experimental) entropy of classical thermodynamics can be naturally understood as causal entropy applied to the special science of classical thermodynamics. This justifies generalizing the term entropy to robust causes and their effects of special sciences for which state-supervenience holds, and shows that it is possible to move beyond the causal second law to practically useful relations between observable quantities describing the behavior of boxes of gases.

The main derivation of \citep[Section IV]{Jaynes1965} for the increase of experimental entropy is known to experts and has been criticized for containing several problematic steps (see, e.g., \citealt[172--174]{Frigg2008}; \citealt[255--257]{Sklar1993}). However, my focus here is not on this main proof, but on Jaynes' intended explication of it (ibid., Section V). This explication has received little philosophical attention.

I start by reconstructing Jaynes' explication as an abstract reasoning independent of his main proof (with a twist regarding how region $R_0$ is defined, explained later), keeping his notation. Suppose that, at an initial time $0$ and a final time $t$, initial and final thermodynamic states are defined by fixing the values of $n$ thermodynamic quantities ${X_1(0), X_2(0), \ldots, X_n(0)}$ and ${X_1(t), X_2(t), \ldots, X_n(t)}$, respectively. Suppose that we can associate the initial thermodynamic state with a phase space region $R_0$ containing all and only those microstates of the phase space that instantiate it in an experimental circumstance; by the same token, we associate another region $R'$ with the final thermodynamic state. Finally, suppose that an experimentally reproducible process takes the first thermodynamic state to the second in such a way that the microphysical equations of motion are phase-volume-preserving.

Under these conditions, Jaynes argues that the Boltzmann entropy of the initial thermodynamic state cannot be larger than the Boltzmann entropy of the final thermodynamic state. Suppose that we prepare the initial thermodynamic state. This is then instantiated by a microstate $p_0$ of the region $R_0$. The underlying physical dynamics $U_t$ takes $p_0$ to the final microstate $p_t = U_t p_0$. If the experimental process takes the initial thermodynamic state to the final thermodynamic state, then $p_t$ must instantiate the final thermodynamic state, that is, $p_t \in R'$. Setting up the same experiment many times, $p_0$ varies across $R_0$, and hence $p_t$ varies across $R_t = U_t R_0$. Thus, if the process taking the initial thermodynamic state to the final thermodynamic state is experimentally reproducible, then for every initial microstate $p_0 \in R_0$, the corresponding final microstate $p_t = U_t p_0$ must instantiate the final thermodynamic state, that is, every $p_t \in R'$. Since we assumed that the equations of motion preserve the phase volume, the phase volume of $R_0$ must be the same as that of $R_t$. Hence, $R_t \subseteq R'$ means that the phase volume $W_0$ of the set of microstates $R_0$ instantiating the initial thermodynamic state cannot exceed the phase volume $W'$ of the set $R'$ instantiating the final thermodynamic state. The scaled logarithms of the phase volumes of thermodynamic states are their Boltzmann entropies. Hence, the Boltzmann entropy of the initial thermodynamic state cannot be larger than the Boltzmann entropy of the final thermodynamic state, which is an expression of the second law of statistical mechanics.

The analogy is straightforward. Defining the initial thermodynamic state by fixing the values of thermodynamic quantities is analogous to defining a cause by causal factors. By the same token, the final thermodynamic state corresponds to an effect. The assumption that the initial and final thermodynamic states can be associated with respective sets of microstates instantiating them is analogous to state-supervenience. A ``reproducible experimental process" connecting the initial and final thermodynamic states is analogous to a robust causal regularity from a cause to an effect. Boltzmann entropy is analogous to causal entropy. Finally, Jaynes' abstract second law for statistical mechanics, interpreted as a law characterizing reproducible experimental processes, is analogous to the causal second law for robust causal regularities.

Jaynes himself suggested that his reasoning extended beyond equilibrium thermodynamics. His reflection is terse enough to permit it to be quoted in full:
\begin{quote}
	It is suddenly clear that the second law is only a very special case of a general restriction on the direction of any reproducible process, whether or not the initial and final states are describable in the language of thermodynamics; the expression $S = k \log W$ gives a generalized definition of entropy applicable to arbitrary nonequilibrium states, which still has the property that it can only increase in a reproducible experiment. \citep[396]{Jaynes1965}
\end{quote}

So far, Jaynes' reconstructed explication does not commit mathematical errors, does not depend on assumptions about initial or boundary conditions of the universe, and, though originating in microphysical dynamics, it leads to an entropy principle that is exceptionless in the sense explained in Section \ref{subsect_tri}. I am not aware of any other foundational approach to the second law that has all these properties. 

However, this cheerful picture comes with at least three caveats. First, the experimental reproducibility condition makes it clear that the derivation does not explain the second law from {\em purely} microphysical premises, which is what some other approaches have attempted to achieve: on the contrary, the derivation explicitly assumes the existence of a reproducible process between macrostates. Second, as I will argue in Section \ref{subsect_fourcausalasymmetries}, one must be careful when interpreting the increase of entropy {\em from cause to effect} as an increase {\em in time}. Jaynes' argument shows that the Boltzmann entropy of the initial thermodynamic state cannot exceed that of the final state; he implicitly assumed that the initial state precedes the final state in time, but this assumption did not play a role in his derivation (as explained in Section \ref{subsect_rovelli}, this becomes less of a caveat once we realize that other approaches to the second law also share this property). Third, the derivation is abstract, and, since it provides no link to equilibrium thermodynamics, its practical value is open to question.

For describing the behavior of boxes of gases, the third problem may be alleviated by reading Jaynes' reasoning in reverse. The aforementioned twist in my reconstruction of Jaynes' argument was that I defined the phase space region $R_0$ so as to have an interpretation that renders the argument valid. Reconstructing the argument this way, the natural question is the following: can Jaynes give another characterization of the region $R_0$ that both satisfies the interpretation of $R_0$ needed for the proof to be valid, and provides a link to experimentally verifiable quantities?

Jaynes' answer would be the following. Suppose that the initial thermodynamic state is in complete thermal equilibrium; experimentally speaking, suppose it was prepared by only controlling for the $n$ thermodynamic quantities, and these quantities are unchanging for a ``sufficiently long" time. We choose the classical $6N$-dimensional phase space with a phase space point representing the positions and momentums of $N$ particles, and assume a classical Hamiltonian $H$ with a symmetric potential energy function depending only on the relative coordinates. Let $W_N(x_1,p_1; x_2,p_2; \ldots, x_N, p_N)$ be the canonical $N$-particle distribution function that agrees with the experimental internal energy $U_e$ of the initial thermodynamic state in the sense that $U_e = \int W_N H d^3 x_1 \cdots d^3p_N$. For $0<\epsilon<1$, define the ``high probability" region $R(\epsilon)$ as consisting of all phase space points for which $W_N \geq C$, where the constant $C$ is chosen so that the total probability of finding the system somewhere in the region $R(\epsilon)$ is $(1-\epsilon)$. Let $W(\epsilon)$ denote the phase volume of $R(\epsilon)$. A theorem by Shannon shows that $\lim_{N \to \infty}((S_G - \log W(\epsilon)) / N) = 0$, independently of the choice of $\epsilon$, where $S_G$ is the Gibbs entropy. Jaynes interprets this theorem as follows: it does not matter how we set the value of $\epsilon$ (what we mean by ``high probability"); if our thermodynamic system contains a sufficiently large number of particles, then the Boltzmann entropy of the region $R(\epsilon)$ equals the Gibbs entropy calculated from the canonical distribution. Thus, we choose a suitable value of $\epsilon$, and let $R_0 = R(\epsilon)$, $W_0 = W(\epsilon)$. In each experiment, the microstate instantiating the experiment is assumed to be selected from the so-defined region $R_0$ with a probability described by $W_N$. The required interpretational link is then established as follows:
\begin{quote}
	The aforementioned variational property of the canonical ensemble now implies that, of all ensembles agreeing with this initial data in the sense of (16) [in the sense that the mean value of the Hamiltonian equals $U_e$], the canonical one defines the {\em largest} high-probability region. The phase volume $W_0$ therefore describes the full range of possible initial microstates; and not some arbitrary subset of them; this is the basic justification for using the canonical distribution to describe partial information. \citep[396]{Jaynes1965}
\end{quote}
These conditions recreate an approximate, probabilistic, but---being constrained by the choice of a classical phase space---physically more realistic version of the state-supervenience assumption that $R_0$ contains all and only those phase space points that instantiate the initial thermodynamic state in some experimental circumstances. 

The approximate agreement between the Boltzmann entropy and the Gibbs entropy, calculated from the canonical distribution, provides the missing experimental link: in complete thermal equilibrium, all reproducible thermodynamic properties are represented by the canonical distribution. In particular, Jaynes argues that the Gibbs entropy equals the so-called {\em experimental entropy} of classical thermodynamics (whose relevant quantities can be introduced as functions of appropriate partial derivatives of the Gibbs entropy). If a similar reasoning could be carried out for the final thermodynamic state, then, since preservation of phase volume can be guaranteed by assuming that the experimental process is adiabatic, and since the non-decrease of the Boltzmann entropy from the initial to the final thermodynamic state for phase-volume-preserving experimental processes was established above, we would arrive at a proof of the non-decrease of experimental entropy between the initial and final states, which is the so-called {\em phenomenological second law} of thermodynamics (for reproducible processes).

Evaluating the merits and pitfalls of this reversed argument for the phenomenological second law lies beyond the scope of this paper and shall be pursued elsewhere. For my present purposes, it is sufficient that this reversed argument shows, through the example of modeling the behavior of boxes of gases, that it is possible to create a plausible and practically useful link between the causal second law and experimentally observable quantities. In particular, Jaynes' reversed argument illustrates how, under familiar assumptions, the practically useful thermodynamic entropy can be seen as causal entropy applied to the special science of thermodynamics.

\subsection{The Reversibility Objection}\label{subsect_tri}

Time reversal invariance of physical dynamics is often thought to contradict the exceptionless character of the second law of classical thermodynamics because of the so-called reversibility objection (for an overview and references, see \citealt{Callender2021}). In this section, I explain the reversibility objection with a simple example, and show why it does not threaten the exceptionless character of the causal second law for robust regularities within the present framework.

For simplicity, let us assume that the physical theory underlying classical thermodynamics is the kinetic theory, according to which the microstate of a gas is fixed by the positions and velocities of its particles. We can introduce a conceptual operation between microstates, the {\em time reversal operator} $T$: given a microstate $p$ as an input, the output is an (almost always different) microstate $q =Tp$, in which every particle has the same position as in $p$ but its velocity points (with the same magnitude) in the opposite direction. 

Suppose that initially a gas is in the microstate $p_0$ in a stationary box A, and after $t^*$ seconds the particles have bounced around, leaving the gas in a new microstate $p_{t^*} = U_{t^*} p_0$. Suppose that, in another stationary box B (a copy of box A), a gas happens to be in the exact time-reversed microstate $q_0 = Tp_{t^*}$, and suppose that we again wait precisely ${t^*}$ seconds for the particles in box B to bounce around, after which the gas in box B arrives to microstate $q_{t^*} = U_{t^*} q_0$. Since the positions of the particles in $q_0$ (relative to box B) were the same as the positions of the particles in $p_{t^*}$ (relative to box A), but their velocities were flipped around, we expect to see all gas molecules in box B to follow a `reverse path', and so ${t^*}$ seconds later all particles should arrive back to the same position (relative to box $B$) as the particles were in $p_0$ (relative to box $A$), but with the opposite velocity. Thus, flipping back the velocities of $q_{t^*}$ should yield the same microstate $p_0$ with which we started in box A. In sum, we expect $p_0 = T^{-1} U_t T U_t p_0$ to hold for every microstate $p_0$ and every time $t \geq 0$. Since non-singular Newtonian collision dynamics indeed has this property, it is called {\em time reversal invariant}.\footnote{For a detailed discussion of time reversal, see \citet{Roberts2022}.}

After this preparation, the reversibility objection runs as follows: if (i) every macrostate has an entropy defined by its phase volume, (ii) the entropy of a microstate equals that of the macrostate it instantiates, (iii) the entropy of a microstate equals that of its time-reversed microstate, and (iv) the physical dynamics is time reversal invariant, then for every trajectory $p_t = U_t p_0$ along which the entropy increases, there exists another trajectory $q_t = U_t q_0$ along which the entropy decreases. Ergo, the second law is either satisfied trivially (entropy always stays constant), or it cannot be exceptionless.\footnote{A usual way so-called Boltzmannian approaches (which embrace conditions (i)-(iii)) respond to the reversibility objection is by accepting that entropy may occasionally decrease, but by arguing that decreases in entropy happen rarely and only for relatively short amounts of time, and hence the second law still holds statistically. Three premises are invoked to arrive at this conclusion: an inverse relationship between closeness of macrostates to being an equilibrium state and their phase volume (strongly suggested by standard calculations in statistical mechanics), some assumption ensuring sufficiently random wandering of the trajectory through phase space (e.g., ergodicity), and the Past Hypothesis (the universe started in a low entropy state; see \citealt{Loewer2007}). \label{footnote_pasthype}}

There are three reasons why the reversibility objection does not apply to the causal second law. Suppose that $C$ is a causal-entropy-increasing robust cause of $E$, and hence the phase volume of $E$ is greater than that of $C$. First, and conclusively, even if we accept assumptions (i)--(iii) of the reversibility objection, and even if we assume that the time-reversed $TC$ (the set of physical states that are time reverses of physical states in $C$) and $TE$ also correspond to special science descriptions, since the time reversal operator preserves the phase volume, it cannot be the case that almost all physical states evolve from $TE$ to $TC$, since $TE$ has a larger phase volume than $TC$, and the phase volume of $TE$ is preserved by the measure-preserving dynamics. Hence, $TE$ cannot be a robust cause of $TC$; therefore, we do not have an exception to the causal second law, which would only demand entropy increase if $TE$ were a robust cause of $TC$.

Second, there is no reason to expect that if $C$ and $E$ correspond to special science descriptions, then $TC$ and $TE$ must also correspond to special science descriptions (even though physical states are closed under time reversal, and thus if $C$ and $E$ are phase space regions, then $TC$ and $TE$ are also phase space regions). Time reversal is a conceptual operation defined on the level of physical states, and even if it were possible to construct a physical operation which implements a piece-wise time reversal of a set of physical states (for our current microscopic theories this is not feasible even for individual microstates, let alone a set of them), there is no reason to expect that such a physical state-level operation has an interpretation as a special science operation under which the set of special science descriptions should be closed. In short, there is no reason to assume that, in general, special science descriptions are closed under time reversal.\footnote{For macrostates (for case (1) descriptions of thermodynamics, see Section \ref{subsect_howspecialsciencedescriptionscomeabout}) this property is frequently assumed to hold. However, there is no reason to assume that this special case generalizes for descriptions of other special sciences (or even to case (3) descriptions of thermodynamics).}

Third, although the causal second law is compatible with (i)--(iii), the proof does not rely on them. As for (i), I associate causal entropy only with causes and effects of special science causal regularities, which form a small subset of all special science descriptions (the mere fact that phase volume is a well defined number for all descriptions does not necessitate its interpretation as causal entropy). As for (ii) and (iii), I associate causal entropy only with special science causes and effects, and not with the microstates instantiating them. Hence, the reversibility objection also does not arise because (i)-(iii) are not necessary for the causal second law.\footnote{Assumptions (i)--(iii) can be naturally incorporated in the present framework, in which case the three premises of Boltzmannian approaches mentioned in footnote \ref{footnote_pasthype} could be similarly invoked to explain the statistical likelihood of a decrease of causal entropy along a microstate trajectory. However, such a decrease would not defeat the causal second law (causal entropy on a microstate trajectory can decrease only when it connects special science descriptions of states of affairs among which there is no robust regularity---for example, when there is no causal connection, or when the causal regularity is merely portional, see Section \ref{sect_motivatingtheassumptions}). See also footnote \ref{footnote_comparetwolaws}.\label{footnote_boltzmannian}}

Thus, the reversibility objection does not threaten the exceptionless character of the causal second law.

\section{Relaxing the Assumptions}\label{sect_motivatingtheassumptions}

\subsection{Supervenience and Measure-Preservation}\label{subsect_relaxingsupervenience}

A straightforward reading of the key assumptions of Section \ref{sect_causalsecondlaw_robust} is metaphysical: there exists a single fundamental physical theory that can be cast as a measure-preserving deterministic dynamical system, and every description of a state of affairs of a special science corresponds to a set of the theory's phase space points. On this reading, the causal second law becomes a statement about the metaphysical fabric connecting the special sciences to fundamental physics.

However, this metaphysical reading can be significantly relaxed without losing the causal second law. The causal second law can be generalized to explicitly allow for multiple realization, in the sense that both the constitution and lawlike behavior of different instantiations of a robust cause and its effect may be described by different dynamical systems (Proposition \ref{prop_multiply_realized_second_law}). As I shall argue elsewhere, the assumption that the underlying physical theory is deterministic may also be relaxed: measure-preserving stochastic systems (a particular type of transition structures of Definition \ref{def_transitionprobabilitystructure}) likewise entail the causal second law.

The one-to-one correspondence of special science descriptions with phase space regions may be relaxed to a correspondence with a ``high probability'' phase space region; in fact, as we have seen, such a relaxation of state-supervenience allows for a plausible inference from the causal second law to the phenomenological second law of thermodynamics. 

Supervenience can also be understood in terms of empirical adequacy instead of a metaphysical correspondence. A special science is said to {\em theory-supervene} on an underlying theory if, whenever the special science assigns different empirically adequate special science descriptions to a pair of space-time regions, the underlying theory also assigns different empirically adequate underlying descriptions to the same pair of space-time regions. For a precise definition of this notion of theory-supervenience and the thesis of {\em empirical structure physicalism} (stating that every current special science theory-supervenes on current and on future physics), which has recently been introduced and defended against various philosophical challenges, the reader is referred to \citep{Gyenis2025b}. Theory-supervenience of a special science on an underlying theory that can be cast as a measure-preserving dynamical system is sufficient to invoke the argument for the causal second law, emphasizing that the resulting notion of causal entropy is relative to the descriptive capabilities of both the special science and the underlying theory and is restricted to the domain of empirical adequacy (in the sense of \citealt[76]{Gyenis2025b}) of the special science in question. The underlying theory, in this context, does not even need to be a theory of physics: in principle, it could be another, intermediary special science as well.

If we understand supervenience of a special science on an underlying dynamical system in terms of theory-supervenience, then empirical adequacy also becomes the cornerstone against which the assumption of measure-preservation of the underlying dynamical system needs to be assessed. Even the most ardent recent critics of the view that closed systems are fundamental agree that an ``accurate characterization of a system's dynamics can be achieved by modeling it as one subsystem of a larger dynamically isolated composite system'' and that the practice of doing so ``has been highly successful in physics'' \citep[2]{CuffaroHartmann2024}. Since dynamically isolated Hamiltonian systems are measure-preserving, if we can embed a special science phenomenon into a larger system that can be treated as a dynamically isolated Hamiltonian system (up to an approximation under which theory-supervenience of the special science on the underlying theory describing the larger system holds), then, for such composites, the corresponding notion of causal entropy can be introduced and the causal second law can be derived. Many special science phenomena could be studied in isolated laboratory environments, and thus we may obtain the respective causal second law; at worst, for the overwhelming majority of known special science phenomena, we can consider the solar system as a sufficiently isolated composite that contains the phenomena as its subsystem. Choosing an unnecessarily large composite may reduce the chance of practical usefulness (Section \ref{subsect_usefulness}), but, contrary to the main worry of  \citep{CuffaroHartmann2024}, we would still not need to treat the whole universe as a closed system to show that the respective causal second law exists.

\subsection{Portional Causal Regularities}\label{subsect_relaxingrobustness}

We can also relax robustness. Robustness requires almost all physical states instantiating the cause to evolve to physical states instantiating the effect. It is natural to relax the ``almost always'' condition to a ``certain portion'': after all, in some cases a burning match and large amounts of flammable material on the floor of a house do not lead to a burned down house, since the fire is extinguished by firefighters who arrive in time (see the analysis of \citet{Fazekas-Gyenis-Szabo-Kertesz2021} of causal preemption). I thus call a causal regularity {\em portional} if an $\alpha$ portion of the physical states instantiating the cause evolve, in a time characteristic of the regularity, to physical states instantiating the effect. $\alpha$ may be called the {\em efficacy} of the portional cause. (For connecting ``portion of phase volume'' with relative frequency of appearance, see Section \ref{subsect_internalcharacterization} and Remark \ref{remark_portional}.)
 
\begin{figure*}
	\begin{center}
		\includegraphics[width=7.5cm]{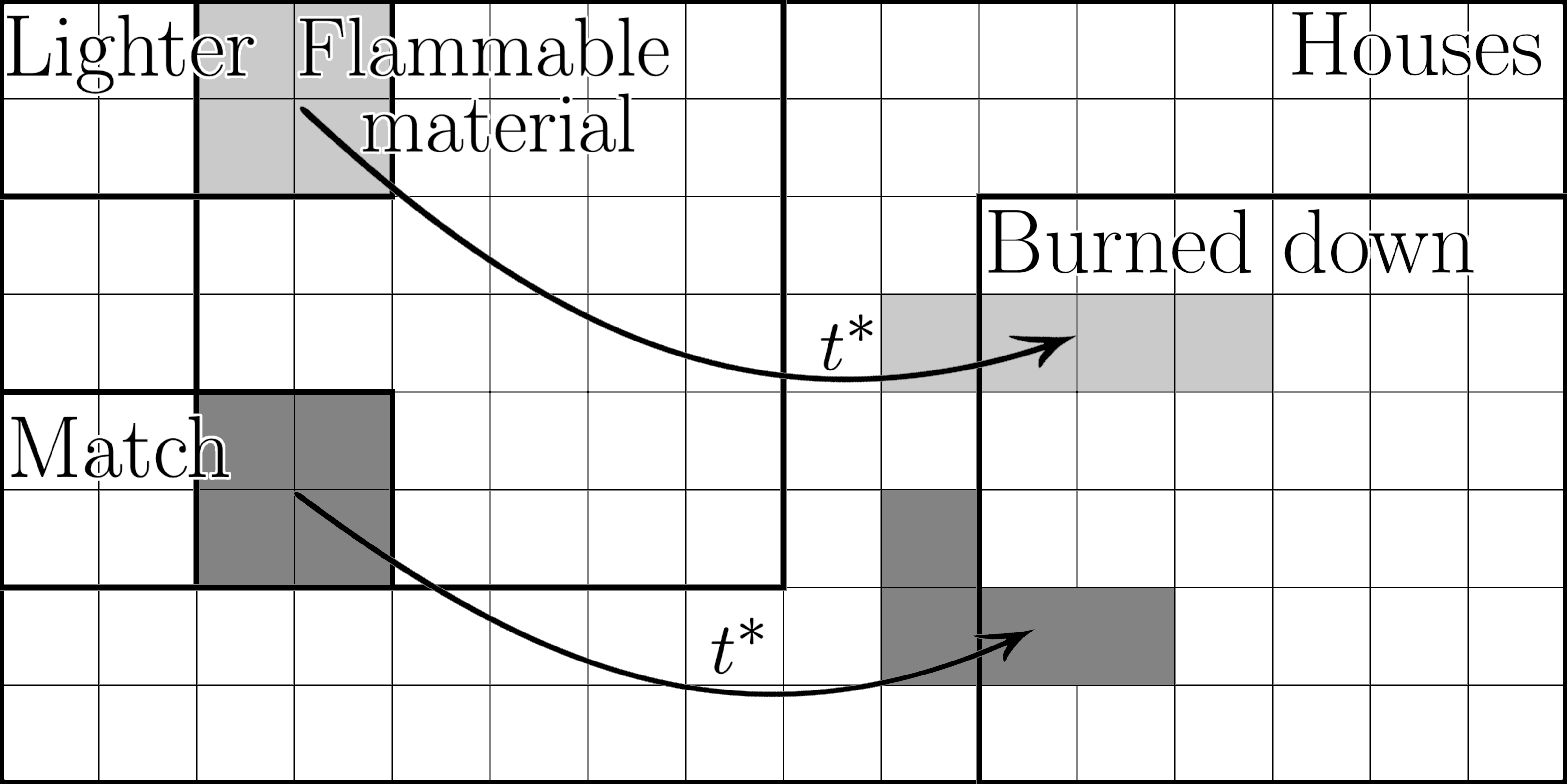}
		\caption{{\em Portional causal regularities with efficacy $\alpha = \frac{3}{4}$ (lit lighter) and $\alpha = \frac{1}{2}$ (burning match).}}
		\label{figure_05}
	\end{center}
\end{figure*}

For a portional causal regularity, measure-preserving dynamics only ensures that the phase volume of the effect is at least as large as an $\alpha$ portion of the phase volume of the cause; thus, in certain cases causal entropy could decrease from cause to effect. Hence, it is important to find conditions that entail that the ``loss'' of a $1-\alpha$ portion of the phase volume is compensated for by other factors. This is why it was important to examine circumstances that guarantee an increase in entropy from cause to effect for a robust causal regularity: the same circumstances may also suffice to entail an increase in entropy from cause to effect for portional causal regularities.

Causal entropy increases for an actual portional causal regularity as long as a sufficient number of distinct possible causes exist for the same effect. A relevant sufficient condition, which depends on the efficacy $\alpha$ and on the phase volume of the distinct possible causes, can be formulated precisely (Proposition \ref{prop_arobust_multiple}). In general, multiplicity of possible causes for the same effect is ubiquitous in the special sciences, and it is rare for the special sciences to formulate causal claims whose efficacy does not meet a reasonable threshold. Hence, it is reasonable to maintain that causal entropy typically increases for portional causal regularities. 

To illustrate, suppose we have two distinct possible causes leading to the same effect with roughly the same phase volume. For instance, suppose that the descriptions `there is a sole burning match in the presence of flammable material in a house' (cause 1) and `there is a sole lit lighter in the presence of flammable material in a house' (cause 2) have the same phase volume. Then, if the two possible causes each lead to the same `burned-down house' effect for at least an average $\alpha = 51\%$ of their instantiating cases, this already guarantees an increase in causal entropy during either of the two particular causal processes (e.g., when an actual lit lighter leads to the house burning down) (see Figure \ref{figure_05} for an illustration).

\section{The Dynamical Systems Approach}\label{sect_dynamicalsystemsapproach}

The set of descriptions of a given special science that satisfies state-supervenience and its underlying physical dynamics jointly determine the pairs of descriptions which stand in a portional causal regularity relationship. What I have called here, for brevity, a portional cause is referred to as a {\em projective state} in \citet{Fazekas-Gyenis-Szabo-Kertesz2021}. However, projective states are only the starting point of their conceptual analysis that aims to show that the dynamical systems approach is both sufficiently flexible to allow for a reconstruction of various kinds of causal claims of the special sciences and everyday life, and that it provides a reductive analysis of these causal claims by outlining truth conditions formulated in terms of a set of descriptions and a physical theory. Their analysis invokes many more conceptual steps to arrive at an account of what {\em ``the'' cause} is (i.e., the so-called causally relevant causal factor that typically separates a so-called principal projective state from a non-projective state along the backwards trajectory-bundle of the most specific special science description that is part of the actual effect). Thus, while descriptions that stand in a portional causal regularity relationship are at the heart of the dynamical systems approach, selecting among several causally relevant factors ``the'' cause, addressing the problem of overdetermination, accounting for causation by absences and misconnections, introducing difference-making intuitions, addressing typical causes, and so on---that is, taking the causal claims of the special sciences and everyday life seriously---requires a more nuanced analysis. \citet{Fazekas-Gyenis-Szabo-Kertesz2021} provide such an informal analysis independently of measure-preservation and without anticipation of an entropy principle. While this paper is self-contained, it can also be understood as a contribution to their approach that formalizes and generalizes some of their core ideas and shows that measure-preservation entails an additional consequence, the causal second law. 

The following sections discuss important areas in which the dynamical systems approach may need further development.

\subsection{Entropy From Cause to Effect vs. Entropy in Time}\label{subsect_fourcausalasymmetries}

Broadly speaking, the dynamical systems approach identifies causation as a portional determination relation between certain descriptions of states of affairs, but, in its current state of development, lacks a principled distinction between determination from the past to the future and from the future to the past. This may be an advantage in the sense that it allows for an exploration of the logical relationship between various causal asymmetries (see below); however, nothing prevents later developments from adopting additional, well-motivated, time-asymmetric requirements for what constitutes a ``cause.'' (I emphasize that this section does not assume that a time reversal operator exists, and is thus conceptually independent of the issue of time reversal, discussed in Section \ref{subsect_tri}.)

With these remarks in mind, let me address how the dynamical causal second law bears on the following four causal asymmetries:
\begin{itemize}
	\item[(1)] The cause precedes its effect in time.
	\item[(2)] The cause leads to its effect.
	\item[(3)] The causal entropy cannot decrease in time.
	\item[(4)] The causal entropy cannot decrease from cause to effect.
\end{itemize}
I claim that, without assuming measure preservation, these four asymmetries are logically independent. Once we assume measure preservation, however, the landscape changes. If we adopt (2) as a conceptual criterion for which member of a timelike determination relationship counts as the cause, then (4) follows from measure preservation, whereas (1) and (3) do not. In fact, for nontrivial cases, (1) holds if and only if (3) does. Thus, the dynamical causal second law yields an asymmetry between cause and effect, but it does not by itself yield a causal time arrow.

As an initial step, I first need to clarify the sense in which (1) can be denied. 

Suppose that two states of affairs of a special science stand in the following relationship: all physical states instantiating the second state of affairs arrive, in a characteristic time, from the set of physical states instantiating the first state of affairs, but many other physical states instantiating the first state of affairs do not evolve to physical states instantiating the second. For such a relationship between the two states of affairs, the first state of affairs occurs {\em earlier} in time; however, while the second state of affairs (together with the physical dynamics) {\em determines} the first state of affairs, the first does not determine the second. Which of the two states of affairs should rather be called ``cause''?

In the context of the dynamical systems approach to causation, it is more natural to identify the second state of affairs as the cause. The approach relies on two key components to define the causal relation: (i) state-supervenience of special science descriptions on sets of physical states, and (ii) the existence of a regular connection between two states of affairs. The choice of which state of affairs should be identified as the cause does not depend on the first key component. However, based on the second key component, it is natural to identify the cause as the state of affairs that determines the other, thus ensuring that causes lead to (determine) their effects. As a consequence of identifying the cause as the state of affairs that determines the effect, the account may allow effects to precede their causes in time, and thus allow (1) to be denied.

\begin{figure*}
	\begin{center}
		\includegraphics[width=7.5cm]{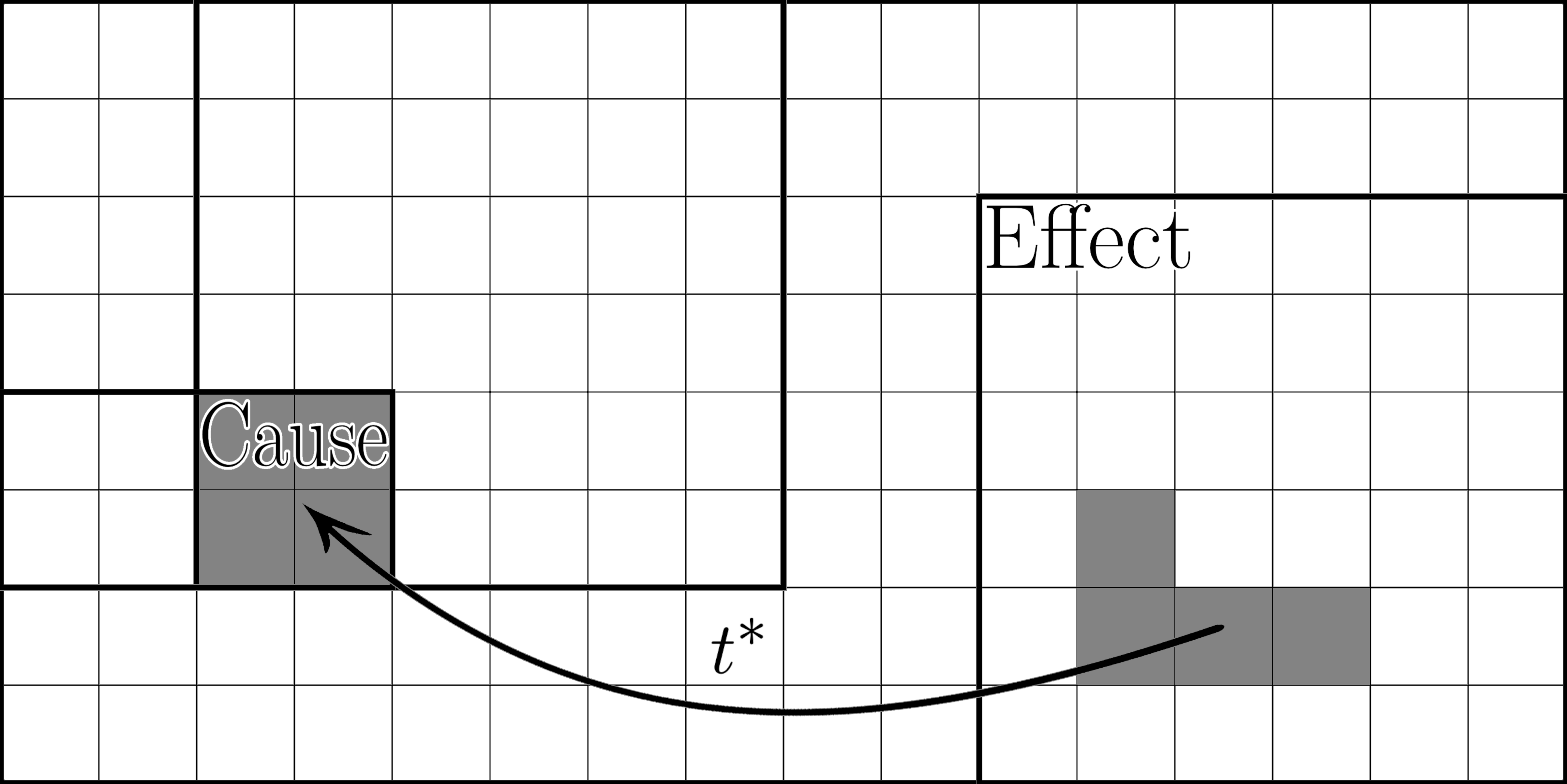}
		\caption{{\em Robust retro-causal regularity: almost all physical states instantiating the cause arrived, in a characteristic time, from the set of physical states instantiating the effect.}}
		\label{figure_06}
	\end{center}
\end{figure*}

I will thus say that the {\em retro-causal regularity} from a cause to an effect is {\em robust} if the backward time evolution takes almost all physical states instantiating the cause into the set of physical states instantiating the effect in a characteristic time (Figure \ref{figure_06}). A {\em robust timelike regularity} is a robust regularity that is causal or retro-causal. Finally, to exclude {\em trivial} cases, I assume that the robust regularity is neither {\em vague} (the effect is not essentially the same as the entire phase space) nor {\em trivially retro-causal} (the causal entropies of the cause and the effect are not the same, since in this case the robust retro-causal regularity would simply be a robust causal regularity with the labeling of the two states of affairs switched) (Definition \ref{def_retrorobust}).

It is easy to show that, if the underlying physical dynamics is measure-preserving, then the entropy of the cause of a robust retro-causal regularity cannot be larger than the entropy of its effect. Thus, combining results, the {\em (dynamical) causal second law} for robust {\em timelike} regularities states that the causal entropy cannot decrease from cause to effect (Proposition \ref{prop_extendedcausalsecondlaw}).

However, unlike the case of a robust causal regularity---where the cause precedes the effect in time and thus the causal entropy cannot decrease in time---in the case of a robust retro-causal regularity, the effect precedes the cause in time, and hence the causal entropy may decrease in time! 

This inference can also be reversed. Suppose measure preservation, and suppose that, for all nontrivial robust timelike regularities expressible by the special science, the causal entropy cannot decrease in time. It then follows that, for all nontrivial timelike regularities expressible by the special science, the cause must precede the effect.

In this sense, under the assumption of measure preservation and nontriviality, if causes precede their effects in time, then the causal entropy cannot decrease in time; conversely, if the causal entropy cannot decrease in time, then causes precede their effects in time.

Thus, the causal second law for robust timelike regularities does not, by itself, guarantee that causal entropy can never decrease in time. Any claim that entropy cannot decrease in time requires the additional, substantive postulate that causes always precede their effects.

If the underlying physical dynamics were not measure-preserving, then the entropy of the cause could be greater than the entropy of the effect for both robust causal and retro-causal regularities. Hence, causal entropy could decrease in time for a robust causal regularity, and it could increase in time for a robust retro-causal regularity. This shows the crucial role that measure preservation plays in connecting the four causal asymmetries.

Finally, (2) can also be relaxed to portionality; and, in fact, robustness also only requires causes to {\em almost} always lead to their effects. As a result, in the current state of development of the dynamical systems approach, the four causal asymmetries are logically independent. 

Although retro-causation receives attention (see \citealt{Faye2021}), the existing philosophical literature does not recognize (3) and (4) as general asymmetries of causal relationships that can be formulated for any special science for which state-supervenience holds, beyond thermodynamics and statistical mechanics. For the same reason, the difference between the third and fourth causal asymmetries has not yet been discussed.

\subsection{Description-Relativity of Causal Entropy}\label{subsect_howspecialsciencedescriptionscomeabout}

I have not provided any {\em a priori} reason why causal regularities exist, or why causes should always precede their effects in time. Whether there exist states of affairs, expressible by a special science, that stand in the relation of a robust causal regularity, or stand in the relation of a robust retro-causal regularity, is determined once the special science descriptions and the underlying dynamical system are fixed. One can easily construct toy examples for which there are nontrivial robust causal regularities but no nontrivial robust retro-causal regularities; other examples for which there are nontrivial robust retro-causal regularities but no nontrivial robust causal regularities; a mix of both; or neither. 

This article leaves open the question how special science descriptions come about; for first steps toward understanding the analogous question for thermodynamics, see \citealt{Hemmo-Shenker2012, Gomori-Gyenis-Szabo2017}. However, I would like to briefly indicate how emphasizing description-relativity of causal regularities and causal entropy may shed light on conceptual issues within thermodynamics as well. 

To start, let me disambiguate three different meanings of ``thermodynamics.'' The set of descriptions of thermodynamics depends on whether by ``thermodynamics'' we mean a theory that (1) specifies a list of thermodynamics {\em quantities} and their relationships, or (2) in addition, also specifies a list of thermodynamic {\em operations} (resulting in changes of constraints), or (3) in addition, also {\em interprets} the thermodynamic operations as physical processes. 

Jaynes focuses his attention on ``thermodynamics'' of case (1), and hence, for his analysis, descriptions of thermodynamics correspond to macrostates: level sets of the chosen thermodynamic quantities obtained for values permitted by their equations of state. Indeed, a gas which, in its initial state, happens to be entirely located in the left half of an isolated box, is a robust cause of a spread-out gas. However, such initial state can neither be properly understood as an equilibrium macrostate, nor is the analysis illuminating from a causal perspective: for instance, we cannot identify which causal factor was ``the'' cause of the gas spreading out, since we did not include the removal of the wall in our set of descriptions. 

While, as this example illustrates, the dynamical systems approach can be applied to ``thermodynamics'' of case (1), it is more illuminating to apply it to case (3),\footnote{For ``thermodynamics'' of case (2), while it is mathematically convenient to treat a difference of thermodynamic constraints of two thermodynamic systems by invoking {\em different} dynamical systems on which the two thermodynamic systems supervene, to properly understand a thermodynamic operation resulting in a {\em change} of constraints, one needs to invoke a {\em single} dynamical system that dynamically implements the operation. But then, we are already at case (3).} especially when the described process can only be treated as a subsystem of a larger, dynamically isolated composite (when the process cannot be treated as dynamically isolated itself, unlike Jaynes' box of gas after the removal of the wall). In case (3), descriptions of thermodynamics correspond to certain subsets of a region of the phase space of a dynamical system that describes the composite in which thermodynamic operations are implemented by physical processes. For instance, for {\em 19th-century thermodynamics} (understood here as the case (3) special science whose thermodynamic quantities stand in the mathematical relations described by classical thermodynamics, but whose descriptions also include those of the usual thermodynamical operations relating to boxes of gases) the underlying dynamical theory may describe an isolated room with a box of gas and a device, and change of volume is implemented by the device removing the wall separating two halves of the box. Thus, descriptions of 19th-century thermodynamics may include, for instance, $C_1$ = `the gas is confined to the left half of the box,' $C_2$ = `the separating wall is removed,' and $E$ = `the gas is spread out to the whole box.' The underlying dynamics entails that the phase space region corresponding to $C_1 \wedge \neg C_2$ is not a robust cause state of $E$, but $C_1 \wedge C_2$ is a robust cause state of $E$; hence, the causal entropy cannot decrease from $C_1 \wedge C_2$ to $E$. Then, the difference making analysis outlined by \citet{Fazekas-Gyenis-Szabo-Kertesz2021} further entails that the causal factor $C_2$ is ``the'' cause of $E$, in agreement with our everyday intuitions about causation. (Statistical mechanics, by itself, provides no such result, since the physical theory lacks a conceptual analysis of which causal factors among those that define a cause state should be identified as ``the'' cause.)

In the closing section of his paper, \citet{Jaynes1965} emphasized the relativity of thermodynamic entropy to a chosen set of thermodynamic quantities. Description-relativity of causal entropy is a general way to express this emphasis: in case (1), description-relativity is equivalent with relativity to a chosen set of thermodynamic quantities, and, in case (3), may also include relativity to a chosen set of thermodynamic operations and to their physical implementation. Thus, in the terminology of this paper, 19th-century thermodynamics is strictly speaking a different special science than thermodynamics applied to black body radiation, degenerate Fermi gases, spin temperature in nuclear-spin systems, or black holes.\footnote{Or, {\em nota bene}, than a thermodynamics whose descriptions additionally include a hypothetically robustly operating Maxwell demon. However, this well-known paradox shall be addressed in detail elsewhere.} What connects these different applications is that their starting point is a state counting notion of entropy (viz. causal entropy), and they all define, for instance, the formal concept ``temperature'' as the inverse of the partial derivative of this entropy with respect to energy, holding other extensive quantities fixed. However, in the different applications the same formal concept ``temperature'' has different empirical interpretations, and the extensive quantities which need to be kept fixed are also differently described: `the wall is not moving and the amount of gas is unchanging' is clearly not the same description as `the angular momentum and the electric charge of the Kerr--Newman black hole are unchanging.'\footnote{Invoking the terminology introduced by \citep{Gyenis2025b}, we may say that these different applications of thermodynamics have the same {\em formal structure} but they have different {\em empirical structures}, and that the empirical structure of all of current thermodynamics is the common refinement of the empirical structures of its currently known applications. Thus thermodynamics, understood as an {\em empirical theory}, changes whenever a new application of it is discovered.}

\subsection{Thermodynamics and Causes Preceding Their Effects}\label{subsect_rovelli}

If we do not have a principled motivation to require a cause to always precede its effect, we may say that, for a given special science, the causal time arrow is {\em real} if it has only nontrivial causal regularities. If nontrivial retro-causal regularities exist, but for some reason cannot be found by practitioners of a special science, then a causal time arrow is only {\em apparent}. As we have seen, the causal second law, by itself, entails neither a real nor an apparent causal time arrow.

If we take the phenomenological second law of 19th-century thermodynamics to be a consequence of the causal second law (in the sense explained in Section \ref{subsect_jaynes}), then, as such it entails neither a real nor an apparent thermodynamic time arrow.  Of course, many versions of the second law have been formulated since the inception of thermodynamics. \citet{Uffink2001}, after an extensive review of the history of different formulations of the second law, concludes that it has nothing to do with the arrow of time; many others (recently: \citealt{Roberts2022}) share Uffink's sentiment. Thus, in this respect, my diagnosis of the causal second law aligns with the conclusions drawn from detailed historical reviews.

This does not mean that a thermodynamic time arrow cannot be established on the basis of independent grounds; it only means that one must be careful in interpreting the increase of entropy from `initial' thermodynamic states to `final' ones as evidence that entropy always increases in time. It would thus be important to understand whether thermodynamics has a real time arrow or, at best, an apparent one.

Many authors (recently: \citealt{Rovelli2023}) have argued that the temporal asymmetry of cause and effect can be derived from the increase of entropy in time, which in turn they assume is provided by thermodynamics. In light of the previous discussion, it is unclear how convincing these arguments can be without a prior, thorough understanding of the source of the assumed thermodynamic {\em time} asymmetry.\footnote{The four causal asymmetries discussed in Section \ref{subsect_fourcausalasymmetries} relate to a fifth asymmetry that is not discussed here, the directionality of the time parameter $t$ of the underlying physical theory. In the philosophy of time literature, several attempts have been made to divest the problem of the directionality of $t$ from causal considerations (for an overview, see \citealt{Callender2017}). My analysis implicitly assumes that the direction in which special science clocks measure time is coordinated with the increasing values of $t$. It would be interesting to see whether one can relax this implicit assumption, but this should be a subject of future work.} If the thermodynamic time arrow is only apparent, it may be the case that even though we have not yet found retro-causal thermodynamic regularities, we may find them in the future. Even if the thermodynamic time arrow is real, it may be the case that, even though retro-causal regularities do not exist on the basis of how thermodynamic properties carve up the phase space, descriptions of another special science carve up the same phase space in a different way such that, for them, retro-causal regularities do exist (and hence, for these special sciences, some effects do precede their causes).

\subsection{The Existence of Special Science Time Arrows}\label{subsect_thermoreduction}

Although, by itself, the causal second law does not provide a causal time arrow, this does not mean that the special sciences cannot motivate, on independent grounds, that their causes always precede their effects. If the special sciences can assume that their causes precede their effects, then the causal second law also entails that causal entropy cannot decrease in time for robust regularities.

Many philosophers (for an overview, see \citealt{Callender2021}) have attempted to reduce the psychological, computational, biological, chemical, and other time arrows to the thermodynamic one, chiefly because they believe that the increase of entropy in time---assumed to be provided by thermodynamics---yields the appropriate type of time asymmetry. However, if a special science has its own causal entropy that increases in time, then such a reduction to thermodynamics is no longer necessary to infer the existence of a special science time arrow, for the sought-after type of entropic asymmetry is already implied by the very existence of a regular connection from the causes to the effects of the special science in question.

\subsection{Transition-Relative-Frequency-Based Characterization of Causal Regularities and the Occurrent Causal Second Law}\label{subsect_internalcharacterization}

Theory-supervenience of special sciences on physics (Section \ref{subsect_relaxingsupervenience}) goes beyond {\em state}-supervenience: it requires differences of special science descriptions of a pair of {\em space-time regions} to be recovered as differences of physics descriptions of the same pair of space-time regions. In general, a description of a space-time region by a theory consists of an assignment of state descriptions to certain moments of time and of the deterministic or probabilistic consequences of these state descriptions that are further entailed by the theory. For a physical theory that is cast as a dynamical system, it is sufficient to assign a single physical state to a moment of time, since the entire trajectory of physical states is further entailed by the time evolution operator. For the special sciences, the further consequences of an assignment of special science descriptions to certain moments of time are, in general, probabilistic, and often can be represented by transition relative frequencies. 

Thus, if we assume that a special science theory-supervenes on a dynamical system, then the dynamical system must also recover the transition relative frequencies entailed by the special science descriptions. In the context of thermodynamics, \citet{Hemmo-Shenker2012} call this latter requirement the ``Probability Rule.''\footnote{It is widely assumed in physics that the symplectic phase-volume measure (which, in canonical local coordinates, is just the Lebesgue measure, which in turn is a generalization of the uniform probability distribution) of a phase space region empirically satisfies this property for fixed levels of energy when no hidden confounding correlations exist during preparation. Finding a philosophically satisfying {\em justification} of this empirical fact is a well-known problem in the foundations of physics; for an excellent analysis, see \citep{Hemmo-Shenker2012}. Future research should address the relationship of this assumption to the question how special science descriptions come about (Section \ref{subsect_howspecialsciencedescriptionscomeabout}).}

By strengthening state-supervenience to {\em history-supervenience} with this additional requirement (Definition \ref{def_supervenes}), we can give an alternative characterization of a robust causal regularity: a causal regularity is robust when the cause leads to the effect with transition relative frequency $1$ in a time characteristic of the regularity (Definitions \ref{def_transitionprobabilitystructure} and \ref{def_robustregularitydoublearrow}, Proposition \ref{prop_connectrobustness}; for portional causal regularities, see Remark \ref{remark_portional}). Since, ideally, transition relative frequencies between special science descriptions of states of affairs may be observed by the special sciences, a transition-relative-frequency-based characterization of a  robust causal regularity is special-science-internal in the sense that it does not explicitly invoke the underlying physical dynamics (Figure \ref{figure_07}).

\begin{figure*}
	\begin{center}
		\includegraphics[width=7.5cm]{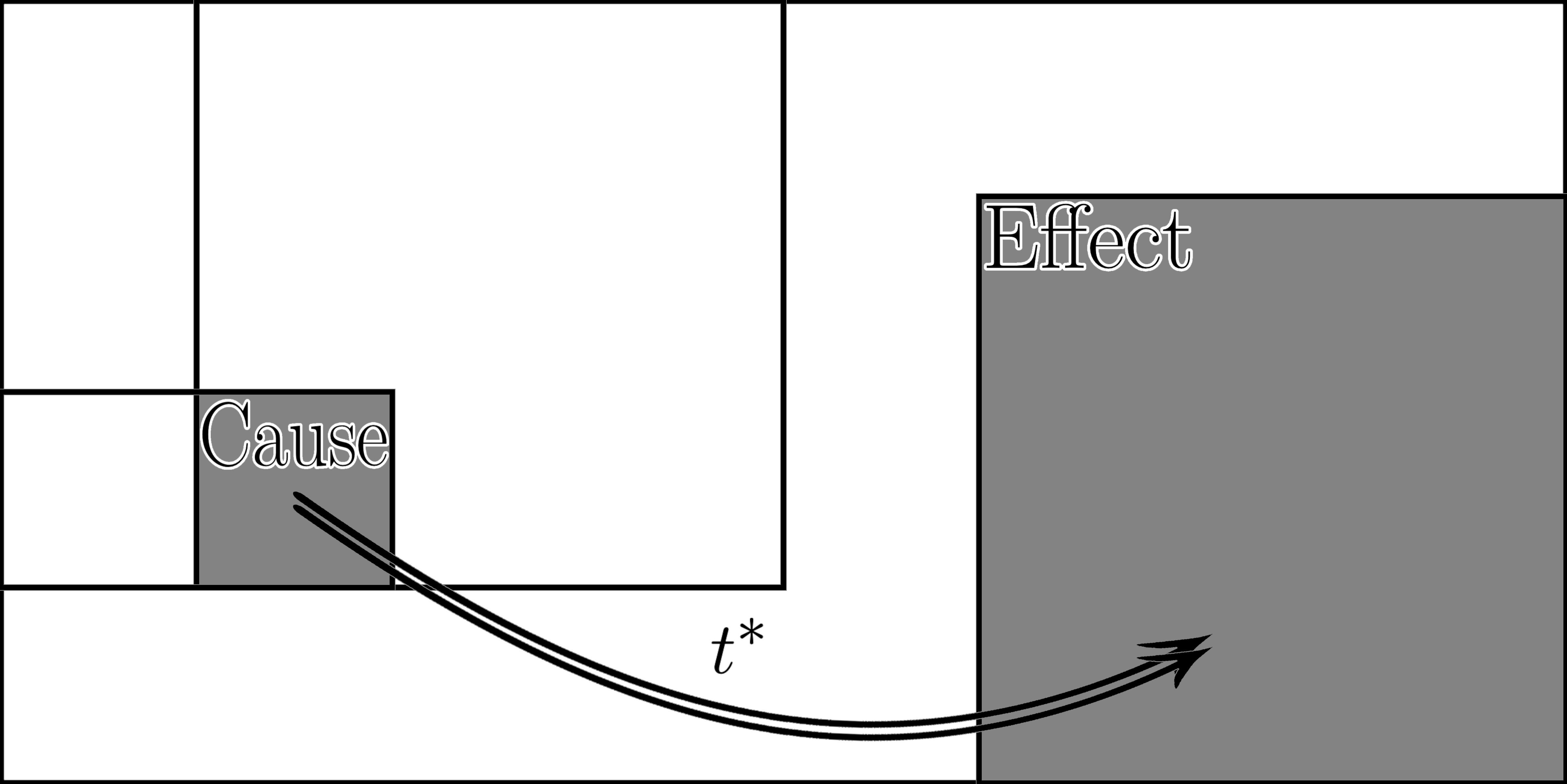}
		\caption{{\em A robust causal regularity, defined with respect to transition relative frequencies, connects special science descriptions, and hence can be illustrated without depicting the physical states corresponding to the cause and the effect.}}
		\label{figure_07}
	\end{center}
\end{figure*}

History-supervenience and measure-preservation entail (Proposition \ref{prop_causalsecondlaw2}):
\begin{description}
	\item The {\bf occurrent causal second law} (for robust causal regularities): robust causes cannot occur more frequently than their effects.
\end{description}

The occurrent causal second law may seem to have an obvious explanation: if the same cause always leads to the same effect, the effect has to occur at least as many times as the cause, hence the frequency of the effect cannot be less than that of the cause! 

However, a puzzle arises about this explanation. We want to contrast how frequently cause states and their effect states occur, but when we count them, we always count cause states that occur earlier than their effect states. Can we simply compare counts taken at different times? Closer inspection reveals that the explanation is only plausible when it is understood as comparing the frequency of {\em histories} instead of the frequency of {\em states}. The number of those histories that end in the effect state indeed cannot be less than the number of histories that end in the effect and start from the cause state, since the latter set of histories is a subset of the former. However, we cannot infer anything further regarding the frequency of cause and effect states without making substantive assumptions about how histories between cause and effect states behave. If histories starting from different instances of the cause state merged before the effect state occurs, then our count would include more instances of cause states than effect states, hence the frequency of the former would be greater than that of the latter!

Thus, the seemingly obvious explanation is misleading: to {\em explain} the occurrent causal second law (as opposed to simply {\em empirically confirm} it), more needs to be assumed. In my analysis, measure-preservation of the underlying physical dynamics is what permits comparing the frequency of occurrences of a robust cause state and its effect, despite that they occur at different times.

Once we characterized the notion of a robust causal regularity with respect to transition relative frequencies, we may turn the inference around: history-supervenience allows for a reductive analysis of a robust causal regularity that is characterized with respect to transition relative frequencies in terms of a robust causal regularity that is defined with respect to a dynamical system (Proposition \ref{prop_connectrobustness}). Illustratively, the double arrow, used in Figure \ref{figure_07} to represent that a cause leads to its effect with relative frequency one, can be replaced by a single arrow, as in Figure \ref{figure_02}, which represents the time evolution that brings almost all physical states instantiating the cause to those instantiating the effect. 

Since it invokes weaker assumptions, the dynamical systems-based characterization of robust causal regularities and the causal second law provide the more general approach (Remark \ref{remark_getridof}), which motivates why this paper focuses primarily on the dynamical systems-based characterization.

\subsection{The Practical Usefulness of the Causal Second Law}\label{subsect_usefulness}

The symbiosis between the discovery of thermodynamics and the development of heat engines is well-known. 19th-century thermodynamics successfully describes and predicts crucial aspects of the behavior of heat engines, and its ability to do so depends only on macroscopic, observable, ``phenomenological'' quantities, without needing to invoke microscopic reduction. However, the success of 19th-century thermodynamics in describing crucial aspects of the behavior of heat engines in a practically useful way leaves open the question of under what conditions this success can be generalized to phenomena other than heat engines.

A heat engine operates as a sequence in which the expansion of a box of gas follows heating, contraction follows cooling, and these states alternate at regular, sharply defined intervals; in our terminology, the successive states of a heat engine form a chain of (nearly) robust cause states and their effects. As Jaynes' reversed argument shows, one can arrive at the practically immensely useful phenomenological second law starting out from this very premise. 

The first step of Jaynes' reversed argument is the claim that reproducible processes entail an abstract second law of statistical mechanics. This paper shows that the first step of Jaynes' reversed argument can be generalized to any special science for which state-supervenience and measure-preservation hold: their causal entropy cannot decrease from a robust cause to its effect. However, the causal second law is only a first step toward practically useful expressions. While my reconstruction suggests that finding (nearly) robust causal regularities is of central importance, an investigation of the conditions under which a reasoning analogous to further steps of Jaynes' reversed argument can also be implemented by particular special sciences---whether some quantities defined from causal entropy correspond to certain special science observables, whether they do so in a way that a subset of these observables can be used to predict some of the remaining ones, or whether they even allow for a full recovery of the formalism of classical thermodynamics---shall be addressed by future research.

\subsection{Closing Remarks}\label{subsect_closingremarks}

The dynamical systems approach is not a regularity account of causation in the sense in which the term is usually understood in the literature: although the approach takes the existence of a regular connection from a cause to an effect to be an important property of the causal relation, it invokes supervenience to anchor this regular connection to physics. The dynamical systems approach differs from inferential accounts, since it does not identify causal claims as inferences made by agents, and since even a robust causal regularity requires only that the cause {\em almost} always lead to the effect, and the notion of a portional regularity significantly relaxes this condition. The dynamical systems approach is not a counterfactual account in the sense the term is usually understood (since it does not formulate the meaning or truth condition of causal claims in terms of counterfactuals).\footnote{The way the dynamical systems approach relativizes causality to different levels of descriptive granularity is similar to \citep{MenziesList2010}, but the latter is a proper counterfactual account on all levels of causation that invokes the Lewisian conception of closest accessible possible worlds to analyze counterfactuals. The dynamical systems approach does not invoke closest accessible possible worlds, yet is able to characterize difference making intuitions (see again \citealt{Fazekas-Gyenis-Szabo-Kertesz2021}). I thank an anonymous referee for calling my attention to \citep{MenziesList2010}.} The approach also makes no reference to dispositions, mechanisms, manipulation, or agency. Although the approach naturally incorporates familiar probabilistic requirements (e.g., by requiring the cause of a portional regularity to increase the relative frequency of its effect), it is not what is usually understood as a probabilistic account, as it does not define the causal relation merely in terms of probabilistic constraints on an algebra or graph. The dynamical systems approach is best understood as a novel, physicalist regularity analysis of causal claims that is substantially different from other physics-based proposals (such as the causal mark transmission and the conserved quantity accounts, see \citealt{Frisch2020}) but is also capable of characterizing the sense in which causation can be regarded as a folk science \citep{Norton2003}.

A better understanding of the aforementioned issues within the dynamical systems approach to causation should be the focus of future research. Analyzing causal claims in terms of dynamical systems is intended as an {\em approach}, rather than as a full-fledged, finalized {\em account}. Future work should provide a more thorough justification of its key assumptions, and relax many simplifications of the approach that were primarily imposed for accessibility. These include a simplified representation of the descriptive capabilities of the special sciences; binary ways of thinking about the existence of causal connections between states of affairs and about the supervenience relation; reliance on a single measure instead of a family of conditional measures (that would be required to properly characterize a cause as a triggering event); the assumption of the existence of a characteristic time, and so on. 

The dynamical systems approach is a new way of thinking about the reductive relationship between the causal claims of the special sciences and physics. I urge readers to contribute to the open problems outlined above and to explore applications of the dynamical systems approach to other philosophical issues.

\begingroup
\small

\section*{Acknowledgements}

I thank Gy\H{o}z\H{o} Egri, M\'arton G\"om\"ori, G\'abor Hofer-Szab\'o, Szabolcs M\'esz\'aros, and two anonymous referees for helpful comments. I also thank audiences of talks for helpful discussion, including those I gave about the dynamical systems approach (starting with an invited talk at the LSE in early 2014). This research has been supported by the OTKA K-134275 and ADVANCED-152165 grants.


\bibliographystyle{apacite}

\begin{thebibliography}{}

\bibitem[\protect\citename{Andreas and Guenther, }2021]{Andreas-Guenther2021}
Andreas, H., and M. Guenther. 
\newblock 2021.
\newblock ``Regularity and Inferential Theories of Causation.''
\newblock In {\em The Stanford Encyclopedia of Philosophy}. Fall 2021, edited by E. N. Zalta. 
\newblock Stanford University Press. 

\bibitem[\protect\citename{Callender, }2017]{Callender2017}
Callender, C. 
\newblock 2017.
\newblock {\em What Makes Time Special?}
\newblock Oxford University Press.
\newblock \url{https://doi.org/10.1093/oso/9780198797302.001.0001} .

\bibitem[\protect\citename{Callender, }2021]{Callender2021}
Callender, C. 
\newblock 2021.
\newblock ``Thermodynamic Asymmetry in Time.''
\newblock In {\em The Stanford Encyclopedia of Philosophy}. Fall 2024, edited by E. N. Zalta. 
\newblock Stanford University Press. 

\bibitem[\protect\citename{Cuffaro and Hartmann, }2024]{CuffaroHartmann2024}
Cuffaro, M. E., and S. Hartmann. 
\newblock 2024.
\newblock ``The Open Systems View.''
\newblock {\em Philosophy of Physics} 2: 1--27.
\newblock \url{https://doi.org/10.31389/pop.90} .

\bibitem[\protect\citename{Faye, }2021]{Faye2021}
Faye, J. 
\newblock 2021.
\newblock ``Backward Causation.''
\newblock In {\em The Stanford Encyclopedia of Philosophy}. Summer 2024, edited by E. N. Zalta. 
\newblock Stanford University Press. 

\bibitem[\protect\citename{Fazekas et~al., }2021]{Fazekas-Gyenis-Szabo-Kertesz2021}
Fazekas, P., B. Gyenis, G. Hofer-Szab\'o, and G. Kert\'esz. 
\newblock 2021.
\newblock ``A Dynamical Systems Approach to Causality.''
\newblock {\em Synthese} 198: 6065--6087.
\newblock \url{https://doi.org/10.1007/s11229-019-02451-y} .

\bibitem[\protect\citename{Frigg, }2008]{Frigg2008}
Frigg, R. 
\newblock 2008.
\newblock ``A Field Guide to Recent Work on the Foundations of Statistical Mechanics.''
\newblock In {\em The Ashgate Companion to Contemporary Philosophy of Physics}, edited by D. Rickles,  99--196.
\newblock Ashgate.
\newblock \url{https://doi.org/10.48550/arXiv.0804.0399} .

\bibitem[\protect\citename{Frisch, }2020]{Frisch2020}
Frisch, M. 
\newblock 2020.
\newblock ``Causation in Physics.''
\newblock In {\em The Stanford Encyclopedia of Philosophy}. Winter 2023, edited by E. N. Zalta. 
\newblock Stanford University Press. 

\bibitem[\protect\citename{G\"om\"ori et~al.,
  }2017]{Gomori-Gyenis-Szabo2017}
G\"om\"ori, M., B. Gyenis, and G. Hofer-Szab\'o. 
\newblock 2017.
\newblock ``How Do Macrostates Come About?''
\newblock In {\em Making it Formally Explicit}, edited by G. Hofer-Szab\'o and L. Wro\'nski, 213--229.
\newblock Springer.
\newblock \url{https://doi.org/10.1007/978-3-319-55486-0} .

\bibitem[\protect\citename{Gyenis, }2025]{Gyenis2025b}
Gyenis, B. 
\newblock 2025.
\newblock ``Empirical Structure Physicalism and Realism, Hempel's Dilemma, and an Optimistic Meta-Induction.''
\newblock {\em Synthese} 206: 76.
\newblock \url{https://doi.org/10.1007/s11229-025-05160-x} .

\bibitem[\protect\citename{Hemmo and Shenker, }2012]{Hemmo-Shenker2012}
Hemmo, M. and O. R. Shenker.
\newblock 2012.
\newblock {\em The Road to Maxwell's Demon}.
\newblock Cambridge University Press.
\newblock \url{https://doi.org/10.1017/CBO9781139095167} .

\bibitem[\protect\citename{Hume, }1739]{Hume1739}
Hume, D. 
\newblock 1739.
\newblock A Treatise of Human Nature.
\newblock \url{https://www.davidhume.org/texts/t/full} .

\bibitem[\protect\citename{Jaynes, }1965]{Jaynes1965}
Jaynes, E. T. 
\newblock 1965.
\newblock ``Gibbs vs Boltzmann entropies.''
\newblock {\em American Journal of Physics} 33: 391--398.
\newblock \url{https://doi.org/10.1119/1.1971557} .

\bibitem[\protect\citename{Loewer, }2007]{Loewer2007}
Loewer, B. 
\newblock 2007.
\newblock ``Counterfactuals and the Second Law.''
\newblock In {\em Causation, Physics, and the Constitution of Reality: Russell's Republic Revisited}, edited by H. Price and R. Corry, 293--326.
\newblock Clarendon Press.

\bibitem[\protect\citename{Menzies and List, }2010]{MenziesList2010}
Menzies, P., and C. List. 
\newblock 2010.
\newblock ``The Causal Autonomy of the Special Sciences.''
\newblock In {\em Emergence in Mind}, edited by C. Macdonald and G. Macdonald.
\newblock Oxford University Press.
\newblock \url{https://doi.org/10.1093/acprof:oso/9780199583621.003.0008} .

\bibitem[\protect\citename{Norton, }2003]{Norton2003}
Norton, J. 
\newblock 2003.
\newblock ``Causation as Folk Science.''
\newblock {\em Philosophers' Imprint} 3: 1--22.

\bibitem[\protect\citename{Roberts, }2022]{Roberts2022}
Roberts, B. 
\newblock 2022.
\newblock {\em Reversing the Arrow of Time}.
\newblock Cambridge University Press.
\newblock \url{https://doi.org/10.1017/9781009122139} .

\bibitem[\protect\citename{Rovelli, }2023]{Rovelli2023}
Rovelli, C. 
\newblock 2023.
\newblock ``How Oriented Causation Is Rooted Into Thermodynamics.''
\newblock {\em Philosophy of Physics} 1: 1--14.
\newblock \url{https://doi.org/10.31389/pop.46} .

\bibitem[\protect\citename{Sklar, }1993]{Sklar1993}
Sklar, L. 
\newblock 1993.
\newblock {\em Physics and Chance: Philosophical Issues in the Foundations of Statistical Mechanics}.
\newblock Cambridge University Press.
\newblock \url{https://doi.org/10.1017/CBO9780511624933} .

\bibitem[\protect\citename{Uffink, }2001]{Uffink2001}
Uffink, J. 
\newblock 2001.
\newblock ``Bluff your way in the Second Law of Thermodynamics.''
\newblock {\em Studies in the History and Philosophy of Modern Physics} 32: 305--394.
\newblock \url{https://doi.org/10.1016/S1355-2198(01)00016-8} .


\end{thebibliography}

\appendix

\section*{APPENDIX}\label{sect_appendixA}
\setcounter{section}{1}

\subsection{Propositions for Section \ref{sect_causalsecondlaw_robust}}

\noindent Suppose $\{d_i\}_{i \in I \subseteq \mathbb{N}}$ are the ``natural'' descriptions of states of affairs that may obtain at a moment in time (e.g., possible causes and effects), given the descriptive capabilities of a particular special science. I assume that, by allowing for countably infinite conjunctions and disjunctions of $d_i$'s and of their negations, we can capture the expressive power of the special science regarding states of affairs at moments in time; the resulting $\sigma$-complete Boolean algebra is isomorphic to a ${\cal D}$ $\sigma$-algebra on a set $D$ by Stone's well-known theorem. State-supervenience (see Definition \ref{def_supervenes}) is the assumption that ${\cal D}$ is isomorphic with a $\sigma$-algebra ${\cal R}$ on an $R$ subset of the phase space $P$ of a physical theory that can be cast as a dynamical system.
\begin{Definition}\label{def_dynamicalsystem}
Let $P$ be a set, $U_t$ a deterministic time evolution operator on $P$ (i.e., for all $t \in \mathbb{R}: U_{t} : P \rightarrow P$ is bijective, and $\{ U_t \}_{t \in \mathbb{R}}$ forms a one-parameter abelian group, hence $U_{t}$ is invertible with notation $U_{-t} \doteq (U_{t})^{-1}$ and for all $t, s \in \mathbb{R}: U_t U_s = U_s U_t = U_{t+s}$), $\Sigma$ a $\sigma$-algebra on $P$ that is closed under $U_t$ images for all $t \in \mathbb{R}$, and $m$ a $\sigma$-finite measure on $(P, \Sigma)$. Then $(P, \Sigma, m, U_t)$ is called a {\em dynamical system}. \\
When it is useful to single out a $0\in T \subseteq \mathbb{R}^{\geq 0}$ set of times, an $R \subseteq P$, and an ${\cal R}$ $\sigma$-algebra on $R$ such that ${\cal R} \subseteq \Sigma$, I write the dynamical system as $(P, \Sigma, R, {\cal R}, T, m, U_t)$.
\end{Definition}

\begin{Remark}
When $(P, \Sigma, R, {\cal R}, T, m, U_t)$ is a dynamical system, since ${\cal R} \subseteq \Sigma$, when $C \in {\cal R}$ I write $C \in {\cal R}$ and $C \in \Sigma$ interchangeably. When ${\cal R}$ and ${\cal D}$ are isomorphic as $\sigma$-algebras, I write $C \in {\cal R}$ and $C \in {\cal D}$ interchangeably.
\end{Remark}

\begin{Remark}
It follows from Definition \ref{def_dynamicalsystem} of a dynamical system $(P, \Sigma, m, U_t)$ that: \\
{\bf (L1)} for all $p\in P, A\subseteq P, t\in \mathbb{R}$: $U_t p \in A \leftrightarrow p \in U_{-t} A$, and hence for all $A, B \subseteq P, t\in \mathbb{R}$: $U_t(A \cup B) = U_t A \cup U_t B$ and $U_t(A \cap B) = U_t A \cap U_t B$. \\
Since $\Sigma$ is closed under $U_t$ images: \\
{\bf (L2)} $U_t A \in \Sigma$ for every $A \in \Sigma, t \in \mathbb{R}$.
\end{Remark}

\begin{Remark}
The set of states in a $C \in \Sigma$ that evolve to states in an $E \in \Sigma$ after $t$ units of time pass is naturally written as $\{ p \in C \mid U_{t} p \in E \}$. Due to (L1), $\{ p \in C \mid U_{t} p \in E \} = \{ p \in C \mid p \in U_{-t} E \} = C \cap U_{-t}E$, and since $\Sigma$ is closed under $U_{-t}$, $C \cap U_{-t}E \in \Sigma$.
\end{Remark}

\begin{Definition}\label{def_robust}
Let $(P, \Sigma, m, U_t)$ be a dynamical system, $C, E \in \Sigma, t^* > 0$. I say that the {\em causal regularity from $C$ to $E$ is robust (with characteristic time $t^*$, and with respect to $(P, \Sigma, m, U_t)$)}---in notation, $C \raisebox{-0.2ex}{$\stackrel{\raisebox{-0.4ex}{$\scriptscriptstyle t^*$}}{\scriptstyle\rightsquigarrow}$} E$---if $m(\{ p \in C \mid U_{t^*} p \in E \}) = m(C)$. 
\end{Definition}

\noindent Let $(P, \Sigma, m, U_t)$ be a dynamical system. The following properties of $\raisebox{0.4ex}{$\scriptstyle\rightsquigarrow$}$ will be used later: \\
\noindent {\bf (L3)} $C \raisebox{-0.2ex}{$\stackrel{\raisebox{-0.4ex}{$\scriptscriptstyle t^*$}}{\scriptstyle\rightsquigarrow}$} E ~~\leftrightarrow~~ m(C \cap U_{-t^*}E) = m(C) ~~\leftrightarrow~~ m(C \cap \neg U_{-t^*} E)=0$. \\ 
\noindent {\bf Proof.} As we have seen, $\{ p \in C \mid U_{t^*} p \in E \} = C \cap U_{-t^*}E$ due to (L1), which entails the first equivalence. Noting that $C$ is the disjoint union of $C \cap U_{-t^*}E$ and $C \cap \neg U_{-t^*}E$, the last equivalence follows from $m(C \cap \neg U_{-t^*} E) = m(C) - m(C \cap U_{-t^*} E)$. \qed \\
\noindent {\bf (L4)} If $C \raisebox{-0.2ex}{$\stackrel{\raisebox{-0.4ex}{$\scriptscriptstyle t^*$}}{\scriptstyle\rightsquigarrow}$} E$, then $(C \cap B) \raisebox{-0.2ex}{$\stackrel{\raisebox{-0.4ex}{$\scriptscriptstyle t^*$}}{\scriptstyle\rightsquigarrow}$} E$ and $C \raisebox{-0.2ex}{$\stackrel{\raisebox{-0.4ex}{$\scriptscriptstyle t^*$}}{\scriptstyle\rightsquigarrow}$} (E \cup B)$, for any $B \in \Sigma$. \\
\noindent {\bf Proof.} If $C \raisebox{-0.2ex}{$\stackrel{\raisebox{-0.4ex}{$\scriptscriptstyle t^*$}}{\scriptstyle\rightsquigarrow}$} E$, then by (L3): $m(C \cap \neg U_{-t^*}E) = 0$, thus $m(C \cap B \cap \neg U_{-t^*}E) = 0$, hence $(C \cap B) \raisebox{-0.2ex}{$\stackrel{\raisebox{-0.4ex}{$\scriptscriptstyle t^*$}}{\scriptstyle\rightsquigarrow}$} E$. \\
If $C \raisebox{-0.2ex}{$\stackrel{\raisebox{-0.4ex}{$\scriptscriptstyle t^*$}}{\scriptstyle\rightsquigarrow}$} E$, then, from (L3) and (L1): $m(C \cap \neg U_{-t^*}(E \cup B)) = m(C \cap \neg (U_{-t^*}E \cup U_{-t^*}B)) = m(C \cap \neg U_{-t^*}E \cap \neg U_{-t^*}B) \leq m(C \cap \neg U_{-t^*}E) = 0$, thus, by (L3), $C \raisebox{-0.2ex}{$\stackrel{\raisebox{-0.4ex}{$\scriptscriptstyle t^*$}}{\scriptstyle\rightsquigarrow}$} (E \cup B)$. \qed \\
\noindent {\bf (L5)} $(C_1 \cup C_2) \raisebox{-0.2ex}{$\stackrel{\raisebox{-0.4ex}{$\scriptscriptstyle t^*$}}{\scriptstyle\rightsquigarrow}$} E ~~\leftrightarrow~~ C_1 \raisebox{-0.2ex}{$\stackrel{\raisebox{-0.4ex}{$\scriptscriptstyle t^*$}}{\scriptstyle\rightsquigarrow}$} E$ and $C_2 \raisebox{-0.2ex}{$\stackrel{\raisebox{-0.4ex}{$\scriptscriptstyle t^*$}}{\scriptstyle\rightsquigarrow}$} E$. \\
\noindent {\bf Proof.} For $\rightarrow$: If $(C_1 \cup C_2) \raisebox{-0.2ex}{$\stackrel{\raisebox{-0.4ex}{$\scriptscriptstyle t^*$}}{\scriptstyle\rightsquigarrow}$} E$, then $m((C_1 \cup C_2) \cap \neg U_{-t^*}E)=0$, thus $m(C_1 \cap \neg U_{-t^*}E) + m((C_2 \setminus C_1) \cap \neg U_{-t^*}E) = 0$, and hence $C_1 \raisebox{-0.2ex}{$\stackrel{\raisebox{-0.4ex}{$\scriptscriptstyle t^*$}}{\scriptstyle\rightsquigarrow}$} E$ (I used (L3); mutatis mutandis for $C_2 \raisebox{-0.2ex}{$\stackrel{\raisebox{-0.4ex}{$\scriptscriptstyle t^*$}}{\scriptstyle\rightsquigarrow}$} E$). \\ 
For $\leftarrow$: if $C_2 \raisebox{-0.2ex}{$\stackrel{\raisebox{-0.4ex}{$\scriptscriptstyle t^*$}}{\scriptstyle\rightsquigarrow}$} E$ then, from (L4), $(C_2 \cap \neg C_1) \raisebox{-0.2ex}{$\stackrel{\raisebox{-0.4ex}{$\scriptscriptstyle t^*$}}{\scriptstyle\rightsquigarrow}$} E$, hence, from $C_1 \raisebox{-0.2ex}{$\stackrel{\raisebox{-0.4ex}{$\scriptscriptstyle t^*$}}{\scriptstyle\rightsquigarrow}$} E$ and (L3): $0 = m(C_1 \cap \neg U_{-t^*}E) + m((C_2 \setminus C_1) \cap \neg U_{-t^*}E)  = m((C_1 \cup C_2) \cap \neg U_{-t^*}E)$, and thus, from (L3), $(C_1 \cup C_2) \raisebox{-0.2ex}{$\stackrel{\raisebox{-0.4ex}{$\scriptscriptstyle t^*$}}{\scriptstyle\rightsquigarrow}$} E$. \qed

\begin{Definition}\label{def_measurepreserving}
A dynamical system $(P, \Sigma, m, U_t)$ is {\em measure-preserving} if, for any time $t \in \mathbb{R}$ and for any $S \in \Sigma :m(U_{-t}S) = m(S)$.
\end{Definition}

\noindent When $(P, \Sigma, m, U_t)$ is measure-preserving, the following two lemmas also hold: \\
\noindent {\bf (L6)} $C \raisebox{-0.2ex}{$\stackrel{\raisebox{-0.4ex}{$\scriptscriptstyle t^*$}}{\scriptstyle\rightsquigarrow}$} E ~~\leftrightarrow~~ m(C) = m(U_{t^*} C \cap E) ~~\leftrightarrow~~ m(U_{t^*} C \cap \neg E) = 0$.\\
\noindent {\bf Proof.} Apply measure-preservation to $m(C) = m(C \cap U_{-t^*} E)$, then use (L1). \qed \\
\noindent {\bf (L7)} If $C \raisebox{-0.2ex}{$\stackrel{\raisebox{-0.4ex}{$\scriptscriptstyle t^*$}}{\scriptstyle\rightsquigarrow}$} E$, then $U_t C \raisebox{-0.2ex}{$\stackrel{\raisebox{-0.4ex}{$\scriptscriptstyle t^*-t$}}{\scriptstyle\rightsquigarrow}$} E$ for every $0 \leq t < t^*$. \\
\noindent {\bf Proof.} If $C \raisebox{-0.2ex}{$\stackrel{\raisebox{-0.4ex}{$\scriptscriptstyle t^*$}}{\scriptstyle\rightsquigarrow}$} E$, then by (L3) $m(C) = m(C \cap U_{-t^*} E)$, and then, by measure-preservation and (L1): $m(U_{t}C) = m(U_{t}(C \cap U_{-t^*} E)) = m(U_{t}C \cap U_{t}U_{-t^*} E) = m(U_{t}C \cap U_{-(t^*-t)} E)$, and thus, by (L3), $U_{t}C \raisebox{-0.2ex}{$\stackrel{\raisebox{-0.4ex}{$\scriptscriptstyle t^*-t$}}{\scriptstyle\rightsquigarrow}$} E$. \qed

\begin{Proposition}\label{prop_causalsecondlaw}
Suppose $(P, \Sigma, m, U_t)$ is a measure-preserving dynamical system, $C, E \in \Sigma$, and $C \raisebox{-0.2ex}{$\stackrel{\raisebox{-0.4ex}{$\scriptscriptstyle t^*$}}{\scriptstyle\rightsquigarrow}$} E$. Then $m(C) \leq m(E)$.
\end{Proposition}
\noindent {\bf Proof.} From measure-preservation and (L6), $m(C) = m(U_{t^*} C) = m(U_{t^*} C \cap E) + m(U_{t^*} C \cap \neg E) = m(U_{t^*} C \cap E) \leq m(E)$. Thus, $m(C) \leq m(E)$. \qed

\subsection{Propositions for Section \ref{sect_strictincrease_robust}}

\begin{Proposition}\label{prop_robust_multiple}
Suppose $(P, \Sigma, m, U_t)$ is a measure-preserving dynamical system, $C_i \raisebox{-0.2ex}{$\stackrel{\raisebox{-0.4ex}{$\scriptscriptstyle t^*_i$}}{\scriptstyle\rightsquigarrow}$} E$, $m(C_i) < \infty$ for $i=1,\ldots,N$, and $m( U_{t^*_i-t^*} C_i \setminus U_{t^*_j-t^*}C_j ) \neq 0$ for every $i \neq j$, where $t^* = \min_i \{ t^*_i \}$. Then $m(E) > m(C_j)$ for all $j = 1,\ldots,N$.
\end{Proposition}
\noindent {\bf Proof.}  From $C_i \raisebox{-0.2ex}{$\stackrel{\raisebox{-0.4ex}{$\scriptscriptstyle t^*_i$}}{\scriptstyle\rightsquigarrow}$} E$ for $i=1,\ldots,N$, it follows from (L5) and (L7) that $\bigcup_i U_{t^*_i-t^*} C_i \raisebox{-0.2ex}{$\stackrel{\raisebox{-0.4ex}{$\scriptscriptstyle t^*$}}{\scriptstyle\rightsquigarrow}$} E$ for $t^* = \min_i \{ t^*_i \}$, and hence, due to Proposition \ref{prop_causalsecondlaw}, $m(E) \geq m(\bigcup_i U_{t^*_i-t^*} C_i)$. But, since $m( U_{t^*_i-t^*} C_i \setminus U_{t^*_j-t^*}C_j ) \neq 0$ for any $i \neq j$, it follows that $m(\bigcup_i U_{t^*_i-t^*} C_i) > m(U_{t^*_j-t^*} C_j) = m(C_j)$ for all $j$, and thus $m(E)  \geq m(\bigcup_i U_{t^*_i-t^*} C_i) > m(C_j)$. \qed

If we assume all characteristic times in Proposition \ref{prop_robust_multiple} to equal, its additional condition entails that $C_i$ and $C_j$ are distinct possible causes, since for $i \neq j$: $m( C_i \setminus C_j ) \neq 0$. In the more general case with different characteristic times, the additional condition expresses that the evolved causes, $U_{t^*_i-t^*} C_i$ and $U_{t^*_j-t^*} C_j$, are distinct when each requires the same remaining time, $t^*$, to produce the effect $E$. 

Let me use the notation $S \approx_{m} T$ for the almost everywhere agreement of $S$ and $T$, that is, when $m((S \setminus T) \cup (T \setminus S))=0$.
\begin{Proposition}\label{prop_robust_mismatch}
There is an example of a $(P, \Sigma, R, {\cal R}, T, m, U_t)$ measure-preserving dynamical system, $C, E \in {\cal R}$, such that  $C \raisebox{-0.2ex}{$\stackrel{\raisebox{-0.4ex}{$\scriptscriptstyle t^*$}}{\scriptstyle\rightsquigarrow}$} E$, but there exists no $S \in {\cal R}$ for which $S \approx_{m} U_{-t^*}E$. \\
In any such case, $m(E) > m(C)$. 
\end{Proposition}
\noindent {\bf Proof.} It is easy to find such examples. For instance, let $P$ be the unit circle, $m$ the uniform distribution on $P$, $U_t$ a clockwise rotation by radians $t, C \doteq \bigcup_{-0.1<t<0.1} U_t(\{ 0 \}), E \doteq \bigcup_{-0.2+1<t<0.2+1} U_t(\{ 0 \})$ (i.e., $C$ and $E$ are disjoint arcs on the circle). Let ${\cal R}$ be the smallest sigma-algebra on $P$ generated by the sets $\{C, E \}$, and let $\Sigma$ be the sigma-algebra generated by ${\cal R}$ and $U_t$. Then $(P, \Sigma, P, {\cal R}, m, U_t)$ is measure-preserving, $C, E \in {\cal R}, C \raisebox{-0.2ex}{$\stackrel{\raisebox{-0.4ex}{$\scriptscriptstyle 1$}}{\scriptstyle\rightsquigarrow}$} E$, and for any $S \in {\cal R}:  S {\not\approx}_{m} U_{-1}E$. Clearly,  $m(E)>m(C)$. \\
If $C \raisebox{-0.2ex}{$\stackrel{\raisebox{-0.4ex}{$\scriptscriptstyle t^*$}}{\scriptstyle\rightsquigarrow}$} E$ for $C, E \in {\cal R}$, but there exists no $S \in {\cal R}$ for which $S \approx_{m} U_{-t^*}E$, then, since $C \in {\cal R}$, we have $C {\not\approx}_{m} U_{-t^*}E$. But since from (L3):  $m(C \setminus U_{-t^*}E) = 0$, it follows from $C {\not\approx}_{m} U_{-t^*}E$ that $m(U_{-t^*} E \setminus C) > 0$. But then from measure-preservation and (L3): $m(E) = m(U_{-t^*} E) = m(U_{-t^*} E \cap C) + m(U_{-t^*} E \cap \neg C) = m(C) + m(U_{-t^*} E \setminus C) > m(C)$. \qed

\subsection{Propositions for Section \ref{sect_motivatingtheassumptions}}\label{app_subsect_motivatingtheassumptions}

\subsubsection{Multiple Realizability}\label{app_subsect_multiple_realizability}

When a special science cause or effect can be realized by physical states from distinct phase spaces, the causal second law extends in a natural way.

\begin{Definition}\label{def_disjoint_union_system}
Let $\{ (P^i, \Sigma^i, m^i, U_t^i) \}_{i \in I}$ be a collection of pairwise distinct dynamical systems indexed by $I = \{1,\dots,N \}$. The {\em disjunctive dynamical system} $(P, \Sigma, m, U_t)$ is defined as follows:
\begin{itemize}
\item $P \doteq \bigcup_{i \in I} \left( \{i\} \times P^i \right)$; \\ For each $i \in I$, let $\iota_i : P^i \to P$ be defined as $\iota_i(x) \doteq (i,x)$;
\item $\Sigma \doteq \left\{ S \subseteq P : \iota_i^{-1}(S) \in \Sigma^i \text{ for all } i \in I \right\}$;
\item $m(S) \doteq \frac{1}{N} \sum_{i \in I} m^i(\iota_i^{-1}(S))$ for every $S \in \Sigma$;
\item $U_t(i,x) \doteq (i, U_t^i(x))$ for every $(i,x) \in P$.
\end{itemize}
\end{Definition}
\noindent It is easy to verify that $\Sigma$ is a $\sigma$-algebra, that $m$ is a $\sigma$-finite measure on $\Sigma$ (which is also a probability measure if all $m^i$ are), and that if each $(P^i, \Sigma^i, m^i, U_t^i)$ is measure-preserving, then $(P, \Sigma, m, U_t)$ is also measure-preserving.

\noindent I call a description $C \in \Sigma$ {\em multiply realized} if it consists of a disjunction of descriptions which state-supervene on different dynamical systems, that is, if $C = \bigcup_{i \in I} \iota_i(C_i)$ with $C_i \in \Sigma^i$ for $i \in I$ (where $C_{i_1} \neq \emptyset \neq C_{i_2}$ for some $i_{1} \neq i_{2}$). This motivates the following definition:
\begin{Definition}\label{def_multiply_realized_robust}
Let $(P, \Sigma, m, U_t)$ be the disjunctive dynamical system constructed from $\{ (P^i, \Sigma^i, m^i, U_t^i) \}_{i \in I}$, and let $C, E \in \Sigma$, with $C = \bigcup_{i \in I} \iota_i(C_i), E = \bigcup_{i \in I} \iota_i(E_i)$ with $C_i, E_i \in \Sigma^i$ for $i \in I$. I say that the {\em multiply realized causal regularity from $C$ to $E$ is robust} if, for each $i \in I$, $C_i \raisebox{-0.2ex}{$\stackrel{\raisebox{-0.4ex}{$\scriptscriptstyle t^*$}}{\scriptstyle\rightsquigarrow}$} E_i$ in the dynamical system $(P^i, \Sigma^i, m^i, U_t^i)$.
\end{Definition}

\begin{Proposition}\label{prop_multiply_realized_second_law}
Suppose $(P, \Sigma, m, U_t)$ is a measure-preserving disjunctive dynamical system, and the multiply realized causal regularity from $C$ to $E$ is robust. Then $m(C) \leq m(E)$.
\end{Proposition}
\noindent {\bf Proof.} Since $C_i \raisebox{-0.2ex}{$\stackrel{\raisebox{-0.4ex}{$\scriptscriptstyle t^*$}}{\scriptstyle\rightsquigarrow}$} E_i$ for each $i \in I$, Proposition \ref{prop_causalsecondlaw} applied to each component system yields $m^i(C_i) \leq m^i(E_i)$. Hence,
\begin{equation*}
m(C) = \frac{1}{N} \sum_{i \in I} m^i(C_i) \leq \frac{1}{N} \sum_{i \in I} m^i(E_i) = m(E). 
\end{equation*}
\qed

\subsubsection{Portional Causal Regularities}\label{app_subsect_portional}

\begin{Definition}\label{def_arobust}
Suppose $(P, \Sigma, R, {\cal R}, T, m, U_t)$ is an underlying dynamical system, $C, E \in \Sigma$, and $m(E) / m(R) < \alpha \leq 1$. I say that the {\em causal regularity from $C$ to $E$ is $\alpha$-portional} (with characteristic time $t^*$, and with respect to $(P, \Sigma, R, {\cal R}, T, m, U_t)$)---in notation, $C \raisebox{-0.2ex}{$\stackrel{\raisebox{-0.4ex}{$\scriptscriptstyle t^*\!,\alpha$}}{\scriptstyle\rightsquigarrow}$} E$---if $m(\{ p \in C \mid U_{t^*} p \in E \}) = \alpha \cdot m(C)$.
\end{Definition}

\begin{Remark}\label{remark_portional}
When $C, E \in {\cal R}$ and history-supervenience (Definition \ref{def_supervenes}, see later) holds, the additional condition can be rewritten as $m(E) / m(R) = \mu_{0,E, R} < \alpha = \mu_{t^*, E, C}$; thus, for $\alpha = 1$, such a regularity is $\alpha$-portional iff it is robust and non-vague (if $\mu_{0,E, R} \neq 1$). The motivation for the additional condition is the familiar requirement that a cause should increase the probability of its effect: $\mu_{0,E,R} < \mu_{t^*, E, C}$. In practice, this additional condition typically holds. First, the special sciences typically maintain a reasonably high threshold for $\alpha$. Second, the unconditional relative frequency of the occurrence of the effect needs to be small ($\mu_{0,E,R} \ll 1$) for the effect to be of any interest. Thus, $\mu_{0,E,R}$ is typically less than any reasonable choice of $\alpha$. (Let me emphasize that the characteristic time $t^*$ is assumed to be finite. When $m(R)=1$ and the underlying dynamics is so-called mixing, the additional condition cannot hold in the limit $t^* \rightarrow \infty$ since $\lim_{t \rightarrow \infty} \mu_{t, E,C} = \lim_{t \rightarrow \infty} m(E \cap U_{t}C)/m(C) = m(E) = m(R \cap U_0 E) = m(R) \mu_{0,E,R} = \mu_{0,E,R}$, where I invoked \eqref{eq_measurerecoversfrequency} in the first and the fourth, and mixing in the second step.)
\end{Remark}

\begin{Proposition}\label{prop_arobust_multiple}
Suppose $(P, \Sigma, R, {\cal R}, T, m, U_t)$ is a measure-preserving underlying dynamical system, and suppose $C_i \raisebox{-0.2ex}{$\stackrel{\raisebox{-0.4ex}{$\scriptscriptstyle t^*\!,\alpha_i$}}{\scriptstyle\rightsquigarrow}$} E$ for $i=1,\ldots,N$, with $C_i, C_j$ disjoint for $i\neq j$ ($C_i, E \in \Sigma$). If, for an $a \in \{1,\dots,N\}: ~\alpha_a m(C_a) \geq m(C_a) - \sum_{i \neq a} \alpha_i m(C_i)$ (in particular: if $\alpha_a \geq 1 - \frac{1}{m(C_a)} \sum_{i \neq a} \alpha_i m(C_i)$ when $m(C_a) \neq 0$), then $m(E) \geq m(C_a)$.
\end{Proposition}
\noindent {\bf Proof.} Since $C_i$ are disjoint and $U_t$ is bijective and measure-preserving, the sets $U_{t^*}C_i$ are disjoint; hence $$m(E) \geq m \left(\bigcup_i (U_{t^*}C_i \cap E)\right) = \sum_i m(U_{t^*}C_i \cap E) = \sum_i \alpha_i m(C_i).$$ If $\alpha_a m(C_a) \geq m(C_a) - \sum_{i \neq a} \alpha_i m(C_i)$ then $\sum_i \alpha_i m(C_i) \geq m(C_a)$, and thus $m(E) \geq m(C_a)$. \qed

\subsection{Propositions for Section \ref{sect_dynamicalsystemsapproach}}\label{app_subsect_fourcausalasymmetries}

\subsubsection{Timelike regularities}

\begin{Definition}\label{def_retrorobust}
Suppose $(P, \Sigma, m, U_t)$ is a dynamical system, $C, E \in \Sigma$. I say that the {\em retro-causal regularity from $C$ to $E$ is robust}---in notation, $C \raisebox{-0.2ex}{$\stackrel{\raisebox{-0.4ex}{$\scriptscriptstyle -t^*$}}{\scriptstyle\rightsquigarrow}$} E$---if, for some $t^*>0$ {\em characteristic time}, $m(\{ p \in C \mid U_{-t^*} p \in E \}) = m(C)$.\\
I say that the {\em timelike regularity from $C$ to $E$ is robust}---in notation, $C \raisebox{-0.2ex}{$\stackrel{\raisebox{-0.4ex}{$\scriptscriptstyle \pm t^*$}}{\scriptstyle\rightsquigarrow}$} E$---if either $C \raisebox{-0.2ex}{$\stackrel{\raisebox{-0.4ex}{$\scriptscriptstyle t^*$}}{\scriptstyle\rightsquigarrow}$} E$, or $C \raisebox{-0.2ex}{$\stackrel{\raisebox{-0.4ex}{$\scriptscriptstyle -t^*$}}{\scriptstyle\rightsquigarrow}$} E$. A robust timelike regularity from $C$ to $E$ is {\em vague} if $m(E) = m(P)$; it is {\em trivially retro-causal} if $m(C) = m(E)$; and it is {\em trivial} if it is vague, trivially retro-causal, or if $m(C)=0$.
\end{Definition}

\begin{Proposition}\label{prop_extendedcausalsecondlaw}
Suppose $(P, \Sigma, m, U_t)$ is a measure-preserving dynamical system, $C, E \in \Sigma$, and $C \raisebox{-0.2ex}{$\stackrel{\raisebox{-0.4ex}{$\scriptscriptstyle \pm t^*$}}{\scriptstyle\rightsquigarrow}$} E$. Then $m(E) \geq m(C)$.
\end{Proposition}
\noindent {\bf Proof.} Immediate generalization of Proposition \ref{prop_causalsecondlaw}.

\subsubsection{State-supervenience, history-supervenience}\label{app_subsubsect_supervenience}

\noindent The following definition of a (partial) transition structure captures the idea that for some times $t$ and for some descriptions $C, E \in {\cal D}$ the relative frequency of $E$ occurring a time $t$ after $C$ occurred is fixed (in a consistent way):
\begin{Definition}\label{def_transitionprobabilitystructure}
Let $D$ be a set, ${\cal D}$ a $\sigma$-algebra on $D$, $0 \in T \subseteq \mathbb{R}^{\geq 0}$. A \emph{total transition structure} is a tuple $(D, {\cal D}, T, \{ \mu_{u} \}_{u \in T \times {\cal D} \times {\cal D}^{-\emptyset}})$ satisfying the following conditions:
\begin{itemize}
    \item[(i)] For all $C, E \in {\cal D}, C \neq \emptyset$ and all $t \in T: ~0\leq \mu_{t, E, C} \leq 1$.
    \item[(ii)] For every fixed $t \in T$ and $C \in {\cal D}^{-\emptyset}$, the map $\mu_{t, \cdot, C}: {\cal D} \rightarrow [0,1]$ defined by $E \mapsto \mu_{t, E, C}$ is a probability measure on ${\cal D}$.
    \item[(iii)] For all $C \in {\cal D}^{-\emptyset}: \mu_{0,C,C}=1$.
    \item[(iv)] \emph{Consistency (Law of Total Probability):} For any $t \in T, E \in \mathcal{D}$, and any disjoint non-empty sets $A, B \in \mathcal{D}$:
    \begin{equation*}
        \mu_{t, E, A \cup B} \cdot \mu_{0, A \cup B, D} = \mu_{t, E, A} \cdot \mu_{0, A, D} + \mu_{t, E, B} \cdot \mu_{0, B, D}.
    \end{equation*}
\end{itemize}
When $(D, {\cal D}, T, \{ \mu_{u} \}_{u \in T \times {\cal D} \times {\cal D}^{-\emptyset}})$ is a total transition structure and $U \subseteq T \times {\cal D} \times {\cal D}^{-\emptyset}$, I call $(D, {\cal D}, T, \{ \mu_{u} \}_{u \in U})$ a {\em (partial) transition structure (on domain ${\cal D}, T$)}.
\end{Definition}

\begin{Definition}\label{def_robustregularitydoublearrow}
Let $(D, {\cal D}, T, \{ \mu_{u} \}_{u \in U})$ be a transition structure, $(t^*, E, C) \in U, t^*>0$. I say that the {\em causal regularity from $C$ to $E$ is robust (with characteristic time $t^*$, and with respect to $(D, {\cal D}, T, \{ \mu_{u} \}_{u \in U})$)}---in notation: $C \raisebox{-0.2ex}{$\stackrel{\raisebox{-0.4ex}{$\scriptscriptstyle t^*$}}{\substack{\rightsquigarrow\\[-0.3em]\rightsquigarrow}}$} E$---if $\mu_{t^*, E, C} = 1$.\end{Definition}

\noindent We gave two characterizations of a robust causal regularity, one with respect to $(D, {\cal D}, T, \{ \mu_{u} \}_{u \in U})$, another with respect to $(P, \Sigma, m, U_t)$. What connects the two characterizations?

\begin{Definition}\label{def_supervenes}
Let $(D, {\cal D}, T, \{ \mu_{u} \}_{u \in U})$ be a partial transition structure and $(P, \Sigma, R, {\cal R}, T, m, U_t)$ a dynamical system. I say that {\em $(D, {\cal D}, T, \{ \mu_{u} \}_{u \in U})$ history-supervenes on $(P, \Sigma, R, {\cal R}, T, m, U_t)$ (on domain ${\cal R}, T$)}, and I call $(P, \Sigma, R, {\cal R}, T, m, U_t)$ the {\em underlying} dynamical system if: 
\begin{itemize}
	\item[(i)] {\em State-supervenience: } ${\cal D}$ is isomorphic to ${\cal R}$ as a $\sigma$-algebra (in particular, $D$ corresponds to $R$ under the isomorphism);
	\item[(ii)] $0 < m(R) < \infty$ (hence, $m(C) < \infty$ for all $C \in {\cal R}$), and $m$ {\em recovers} $\{ \mu_{u} \}_{u \in U}$ in the sense that for all $(t, E, C) \in U$ (note: by definition, $E,C \in {\cal D}$ and hence, under the isomorphism assumed in (i), $E, C \in {\cal R}$) the so-called {\em Probability Rule} holds: 
\begin{equation}\label{eq_measurerecoversfrequency}
\mu_{t, E, C} \cdot m(C) = m(U_{-t}E \cap C).
\end{equation}
\end{itemize}
\end{Definition}

\begin{Remark}\label{remark_inferredtransitionstructure}
Let $(P, \Sigma, R, {\cal R}, T, m, U_t)$ be a dynamical system, $t_0 \in \mathbb{R}$, and $0 < m(C) < \infty$ for all $\emptyset \neq C\in {\cal R}$. $(P, \Sigma, R, {\cal R}, T, m, U_t)$ then implies, through equation \eqref{eq_measurerecoversfrequency}, the (up to a $\sigma$-algebra isomorphism) unique total transition structure $(R, {\cal R}, \mathbb{R}^{\geq 0}, \{ \mu_{u} \}_{u \in \mathbb{R}^{\geq 0} \times {\cal R} \times {\cal R}^{-\emptyset}})$ that history-supervenes on $(P, \Sigma, R, {\cal R}, T, m, U_t)$. However, Definition \ref{def_supervenes} is not superfluous; it emphasizes the idea that we may have independent access to an empirical transition structure (obtained from observations of transition relative frequencies of descriptions in ${\cal R}$), and the question is whether the implied $ \{ \mu_{u} \}_{u \in \mathbb{R}^{\geq 0} \times {\cal R} \times {\cal R}^{-\emptyset}}$ agrees with this empirical transition structure (wherever the latter is defined).
\end{Remark}

\begin{Proposition}\label{prop_connectrobustness}
Let $(D, {\cal D}, T, \{ \mu_{u} \}_{u \in U})$ be a transition structure that history-supervenes on the dynamical system $(P, \Sigma, R, {\cal R}, T, m, U_t)$, and let $(t^*, E, C) \in U$ (correspondingly, $E, C \in {\cal R}$), $m(C)>0$. Then $C \raisebox{-0.2ex}{$\stackrel{\raisebox{-0.4ex}{$\scriptscriptstyle t^*$}}{\substack{\rightsquigarrow\\[-0.3em]\rightsquigarrow}}$} E$  iff $C \raisebox{-0.2ex}{$\stackrel{\raisebox{-0.4ex}{$\scriptscriptstyle t^*$}}{\scriptstyle\rightsquigarrow}$} E$.
\end{Proposition}
\noindent {\bf Proof.} Assume the conditions. If $C \raisebox{-0.2ex}{$\stackrel{\raisebox{-0.4ex}{$\scriptscriptstyle t^*$}}{\substack{\rightsquigarrow\\[-0.3em]\rightsquigarrow}}$} E$, then $\mu_{t^*, E, C} = 1$, hence $1 \cdot m(C) = m(U_{-t^*}E \cap C) = m(\{ p \in C \mid U_{t^*} p \in E \})$, and thus $C \raisebox{-0.2ex}{$\stackrel{\raisebox{-0.4ex}{$\scriptscriptstyle t^*$}}{\scriptstyle\rightsquigarrow}$} E$. Conversely, if $C \raisebox{-0.2ex}{$\stackrel{\raisebox{-0.4ex}{$\scriptscriptstyle t^*$}}{\scriptstyle\rightsquigarrow}$} E$, then $m(C) = m(U_{-t^*}E \cap C)$, and since $m(C)>0$ (and $m(C) < \infty$ due to $C \in {\cal R}, m(R) < \infty$), $1 = \frac{m(U_{-t^*}E \cap C)}{m(C)} = \mu_{t^*, E, C}$, hence $C \raisebox{-0.2ex}{$\stackrel{\raisebox{-0.4ex}{$\scriptscriptstyle t^*$}}{\substack{\rightsquigarrow\\[-0.3em]\rightsquigarrow}}$} E$. \qed

\begin{Remark}\label{remark_getridof}
$C \raisebox{-0.2ex}{$\stackrel{\raisebox{-0.4ex}{$\scriptscriptstyle t^*$}}{\substack{\rightsquigarrow\\[-0.3em]\rightsquigarrow}}$} E$ is only defined for $C, E \in {\cal D}$ (correspondingly, for $C, E \in {\cal R}$) and for $t^* \in T$, while $C \raisebox{-0.2ex}{$\stackrel{\raisebox{-0.4ex}{$\scriptscriptstyle t^*$}}{\scriptstyle\rightsquigarrow}$} E$ may also hold for some $C, E \in \Sigma$ for which $C, E \not\in {\cal R}$ and for $t^* \in \mathbb{R}$. Hence, even though $\substack{\rightsquigarrow\\[-0.3em]\rightsquigarrow}$ has the conceptual advantage of being defined without assuming history-supervenience, $\raisebox{0.4ex}{$\scriptstyle\rightsquigarrow$}$ provides a mathematically more general definition of a robust causal regularity from which relevant properties of $\substack{\rightsquigarrow\\[-0.3em]\rightsquigarrow}$ can be inferred when history-supervenience is assumed (see Remark \ref{remark_inferredtransitionstructure}). This motivates working mathematically only with $\raisebox{0.4ex}{$\scriptstyle\rightsquigarrow$}$ and dropping references to the implied $\substack{\rightsquigarrow\\[-0.3em]\rightsquigarrow}$. In the same spirit, I omitted introducing the respective $\substack{\rightsquigarrow\\[-0.3em]\rightsquigarrow}$ for various generalizations of $\raisebox{0.4ex}{$\scriptstyle\rightsquigarrow$}$ in Sections \ref{app_subsect_motivatingtheassumptions} and \ref{app_subsect_fourcausalasymmetries}.
\end{Remark}

\begin{Proposition}\label{prop_causalsecondlaw2}
Suppose $(P, \Sigma, R, {\cal R}, T, m, U_t)$ is a measure-preserving dynamical system, $(R, {\cal R}, T, \{ \mu_{u} \}_{u \in U})$ history-supervenes on $(P, \Sigma, R, {\cal R}, T, m, U_t)$, $(0,C,R) \in U$, $(0,E,R) \in U$, and $C \raisebox{-0.2ex}{$\stackrel{\raisebox{-0.4ex}{$\scriptscriptstyle t^*$}}{\scriptstyle\rightsquigarrow}$} E$. Then $\mu_{0,C,R} \leq \mu_{0,E,R}$.
\end{Proposition}
\noindent {\bf Proof.} Assume the conditions. From Proposition \ref{prop_causalsecondlaw}, $m(C) \leq m(E)$. When $(R, {\cal R}, T, \{ \mu_{u} \}_{u \in U})$ history-supervenes on $(P, \Sigma, R, {\cal R}, T, m, U_t)$ and $(0,C,R), (0,E,R) \in U$, eq. \ref{eq_measurerecoversfrequency} entails $m(C) = \mu_{0,C,R} m(R), m(E) = \mu_{0,E,R} m(R)$, and since $0< m(R) <\infty$, hence $\mu_{0,C,R} \leq \mu_{0,E,R}$. \qed

\endgroup

\end{document}